\documentclass[11pt,centertags,reqno]{amsart}

\usepackage{latexsym}
\usepackage[english]{babel}
\usepackage[T1]{fontenc}
\usepackage{amssymb}
\usepackage{fancyhdr}
\usepackage{url}
\usepackage{hyperref}
\usepackage{verbatim}

\numberwithin{equation}{section} \makeatletter
\renewcommand{\subsection}{\@startsection
{subsection}{2}{0mm}{\baselineskip}{-0.25cm}
{\normalfont\normalsize\bf}} \makeatother

\newtheorem{theorem}{Theorem}[section]
\newtheorem{lemma}[theorem]{Lemma}

\newtheorem{definition}[theorem]{Definition}
\newtheorem{remark}[theorem]{Remark}
\newtheorem{proposition}[theorem]{Proposition}
\newtheorem{example}[theorem]{Example}
\newtheorem{ass}[theorem]{Assumption}
\newtheorem{notation}[theorem]{Notation}
\def \F {\mathcal F}

\def \H {\mathcal H}
\def \C {\mathcal C}
\def \V {\mathcal V}
\def \N {\mathcal N}
\def \L {\mathcal L}
\def \P {\mathbf P}
\def \Q {\mathbf Q}
\def \R {\mathbb R}

\def \I {{\mathbf 1}}
\def \bF {\mathbb F}

\def \bH {\mathbb H}
\def \bE {\mathbb E}

\newcommand{\ud}{\mathrm d}
\newcommand{\ds}{\displaystyle}
\newcommand{\esp}[2][\mathbb E] {#1\left[#2\right]}

\newcommand{\condespf}[2][\F_t]       {\mathbb E\left.\left[#2\right|#1\right]}

\newcommand{\condesphto}[2][\H_{t_0}]       {\mathbb E\left.\left[#2\right|#1\right]}
\newcommand{\condesph}[2][\H_t]       {\mathbb E\left.\left[#2\right|#1\right]}
\newcommand{\condespho}[2][\H_{0}]       {\mathbb E\left.\left[#2\right|#1\right]}

\author[C.~Ceci]{Claudia  Ceci}
\author[K.~Colaneri]{Katia Colaneri}
\author[A.~Cretarola]{Alessandra Cretarola}

\begin{document}
\address{Claudia  Ceci, Department of Economics,
University ``G. D'Annunzio'' of Chieti-Pescara, Viale Pindaro, 42,
I-65127 Pescara, Italy.}\email{c.ceci@unich.it}

\address{Katia Colaneri, Department of Economics,
University ``G. D'Annunzio'' of Chieti-Pescara, Viale Pindaro, 42,
I-65127 Pescara, Italy.}\email{katia.colaneri@unich.it}

\address{Alessandra Cretarola, Department of Mathematics and Computer Science,
 University of Perugia, via Vanvitelli, 1, I-06123 Perugia, Italy.}\email{alessandra.cretarola@dmi.unipg.it}

\title{A Benchmark Approach to Risk-Minimization under Partial Information}

\date{}

\begin{abstract}
\begin{center}
In this paper we study a risk-minimizing hedging problem for a semimartingale incomplete financial market where d+1 assets are traded continuously and whose price is expressed in units of the num\'{e}raire portfolio. According to the so-called {\em benchmark approach}, we investigate the (benchmarked) risk-minimizing strategy in the case where there are restrictions on the available information. More precisely, we characterize the optimal strategy as the integrand appearing in the Galtchouk-Kunita-Watanabe decomposition of the benchmarked claim under partial information and provide its description in terms of the integrands in the classical Galtchouk-Kunita-Watanabe decomposition under full information via dual predictable projections.
Finally, we apply the results in the case of a Markovian jump-diffusion driven market model where the assets prices dynamics depend on a stochastic factor which is not observable by investors.
\end{center}
\end{abstract}

% 91G10-Portfolio theory 91G20-Derivative securities
% 60G46-Martingales and classical analysis 60G55-Point processes 60G57-Random measures
% 60J25-Continuous time Markov processes 60J60- Diffusion processes 60J75-Jump processes


\maketitle

\section{Introduction}

In this paper we investigate a risk-minimizing hedging problem under the so-called {\em benchmark approach} (see e.g.~\cite{p2005} and~\cite{ph}) for derivatives in an incomplete financial semimartingale market model where there are restrictions on the information available to traders.
Furthermore, we apply the results to discuss the case of a Markovian jump-diffusion driven market model. \\
Here the incomplete information framework is modeled by the presence of two filtrations, the first one denoted by $\bF$ representing the full knowledge on the market and the other $\bH$, supposed to be smaller than $\bF$, standing for the available information level.
Partially observable models often describe realistic financial scenarios. A common example is given by the case where investors may observe the underlying dynamics only at fixed times or when the market is influenced by an unobservable stochastic factor that may represent the trend of a correlated market.

Concerning the risk-minimizing approach (see e.g.~\cite{fs86},~\cite{fs1991} and~\cite{sch01}), in the case of restricted information there are few results in the current literature, as far as we are aware.
Important contributions in this direction, performed with the savings account as reference unit, can be found in~\cite{s94} and ~\cite{ccr}.
In~\cite{s94}, the author provides an explicit expression for risk-minimizing hedging strategies under restricted information in terms of predictable dual projections, whereas in~\cite{ccr}, by proving a version of the Galtchouk-Kunita-Watanabe decomposition that works under partial information, the authors extend the results of~\cite{fs86} to the partial information framework and show how their result fits in the approach of~\cite{s94}. All the above-mentioned papers deal with the case where the risky assets prices are modeled directly under a risk-neutral measure. The general semimartingale case, that corresponds to assume that the risky assets behavior is described by a semimartingale,  is more difficult to handle even under full information and the weaker concept of local risk-minimization must be introduced. In particular in~\cite{fs1991} it is proved that the existence of an optimal strategy is equivalent to the
existence of the so-called {\em F\"{o}llmer-Schweizer} decomposition of the contingent claim. This decomposition is, in a certain sense, a generalization of the Galtchouk-Kunita-Watanabe decomposition that holds when the underlying price process is a local martingale.
In the case where the assets prices dynamics have continuous trajectories, the problem can be solved with the aid of the {\em minimal martingale measure} (see e.g. \cite{sch01}).\\
A first step about the study of local risk-minimization under partial information is represented by~\cite{fs1991} where the authors complete the information starting from the reference filtration and recover the optimal strategy by means of predictable projections with respect to the enlarged filtration that makes the market complete. Furthermore, an application of the local risk-minimizing approach in the case of incomplete information to defaultable markets in the sense of~\cite{fs1991} can be found in~\cite{bc}. Another contribution in this direction is given by~\cite{ccr2} where the authors   provide a version of the F\"ollmer-Schweizer decomposition of a square-integrable random variable (that typically represents the payoff of a contract) with respect to the underlying price process, that works under partial information and study the relationship between this decomposition and the existence of a locally risk-minimizing strategy according to the partial information framework.

In this paper, we discuss the risk-minimizing problem for partially observable semimartingale models under the benchmark approach. According to this modeling framework, even under the absence of an equivalent local martingale measure, contingent claims can be consistently evaluated by means of the so-called {\em real-world pricing formula}, see \eqref{realwpf}, which generalizes standard valuation formulas, where the discounting factor is the num\'{e}raire portfolio and the pricing measure is the physical probability measure $\P$.
Risk-minimization under the benchmark approach has been also studied in~\cite{bcp} for general semimartingale models
and in~\cite{dp} where the example of jump-diffusion driven market models is also discussed. However, both of the papers provide results working in the case of complete information. Here, we study the problem in a restricted information setting, which will correspond to {\em benchmarked risk-minimization under partial information}. One of our main achievements concerns the characterization of the optimal strategy via the Galtchouk-Kunita-Watanabe decomposition holding under partial information, see Theorem \ref{prop:fs}.

Most of the literature concerning pricing and hedging of contingent claims focuses on the cases where the underlying asset prices dynamics are given by diffusions or pure jump processes. The latter are usually employed to model high frequency data. Jump driven market models, in fact, are useful to deal with intraday information, when assets prices can be assumed to be piecewise constant and jumps occur when new information arrive (see e.g.~\cite{C},~\cite{cg},~\cite{cg1},~\cite{frey2000},~\cite{fr2001} and~\cite{rs2000}).
Nevertheless, most of the real situations can be covered by jump-diffusion market models. This case allows prices dynamics to have piecewise continuous trajectories but they may present larger jumps, see~\cite{BKR} and~\cite{r2002}. In the models discussed in these papers, the jump part is described by point processes.
Here we analyze the continuous time jump-diffusion driven market model where in particular the jump part in the underlying assets prices dynamics is given by a general random measure. More precisely,  we compute explicitly the benchmarked risk-minimizing strategy under partial information in a Markovian setting (see Proposition \ref{prop:strategia_esplicita}) and characterize the optimal portfolio value in terms of a suitable function solving a partially differential equation (in short PDE) with final condition. Finally, to establish a link with the classical asset pricing theory we investigate the relationship between the the num\'{e}raire portfolio and the existence of a martingale measure for a general jump-diffusion driven market model.

%The local risk-minimization approach is a quadratic hedging method for contingent claims in incomplete markets which keeps the replication constraint and looks for a hedging strategy (in general not self-financing) with minimal cost, see e.g.~\cite{fs} and~\cite{s01} for a further discussion on this issue.

The paper is organized as follows. In Section \ref{sec:numeraire} we describe the financial market in the benchmark perspective. Then, in Section \ref{sec:brm} we study the benchmarked risk-minimization under partial information and provide an explicit representation of the optimal strategy. An application to the case of a Markovian jump-diffusion driven market model can be found in Section \ref{sec:jumpdiff}. Finally, in Section \ref{sec:rel} we discuss the relationship between the num\'{e}raire portfolio and the existence of a martingale measure for the underlying market model.

\section{Benchmark Framework}\label{sec:numeraire}

We fix a probability space $(\Omega,\F,\P)$ endowed with a filtration $\bF:=\{\F_t, \ t\in [0,T]\}$ that satisfies the usual conditions of right-continuity and completeness, where $T \in (0,+\infty)$ is a fixed and finite time horizon. Moreover, we assume that $\F=\F_T$. We consider a financial market model with $d$ $\bF$-adapted, nonnegative primary
security account processes $S^j=\{S_t^j,\ t \in [0,T]\}$, $j \in \{1,2,\ldots,d\}$,
$d\in \{1,2,\ldots\}$. In addition, the $0$-th security
account $S_t^0$ denotes the value of the $\bF$-adapted strictly positive
savings account at time $t \in [0,T]$.
The $j$-th {\em primary security account} holds units of the $j$-th primary security plus its accumulated
dividends or interest payments, $j \in
\{1,2,\ldots, d\}$.
In this setting, market participants can trade in order to reallocate
their wealth.

\begin{notation}
In the sequel we will put
\[
\int X_s \ud Y_s := \left\{\int_0^t X_s \ud Y_s, \ t \in [0,T] \right\}
\]
to identify the stochastic integral process
for every pair $(X,Y)$ of $\R$-valued stochastic processes such that the above expression is well-defined and we write
\[
\int X_s^\top \cdot \ud Y_s := \left\{\int_0^t X_s^\top \cdot \ud Y_s, \ t \in [0,T] \right\}
\]
if $X$ is a vector-valued process and $Y$ is a vector-valued or a matrix-valued process. Here $a^\top$ identifies the transpose of the vector-valued process $a$ and the symbol $\int_0^t X_s^\top \cdot \ud Y_s$ means the sum of the integrals componentwise, i.e. $\sum_i\int_0^t X_s^i \cdot \ud Y^i_s$, if $Y$ is a vector-valued process and $\sum_i\int_0^t X^i_s \cdot \ud Y^{i,j}_s$ for every $j$, if $Y$ is a matrix-valued process.
\end{notation}

\begin{definition}
We call  {\em strategy} an $\R^{d+1}$-valued process
$\delta=\{\delta_t=(\delta_t^0,\delta_t^1,\ldots,\delta_t^d)^\top$,
$t\in [0,T]\}$,
where for each $j \in
\{0,1,\ldots,d\}$, the process $\delta^j=\{\delta_t^j,\ t \in[0,T]\}$
is $\bF$-predictable and integrable with respect to $S^j$.
\end{definition}

\noindent As usual
$\delta_t^j$, for every $j \in \{0,...,d\}$, indicates the number of units of the $j$-th
security account that are held at time $t \in [0,T]$ in the associated
portfolio whose value is given by a c\`{a}dl\`{a}g $\bF$-adapted process $S^\delta=\{S_t^\delta,\ t \in
[0,T]\}$ such that
$$
S_{t^-}^\delta:=\delta^\top_t \cdot S_{t^-}=\sum_{j=0}^d \delta_t^j S_{t^-}^j, \quad t \in [0,T],
$$
where $S=\{S_t=(S_t^0,S_t^1,\ldots,S_t^d)^\top,\ t
\in [0,T]\}$.
A strategy $\delta$ and the corresponding %benchmarked
portfolio $S^\delta$ are
said to be {\em self-financing} if
$$
S_t^\delta=S_0^\delta+\int_0^t\delta^\top_u \cdot \ud S_u = S_0^\delta + \sum_{j=0}^d \int_0^t\delta_u^j\ud S_u^j, \quad t \in [0,T].
$$
%where
%$\delta=\{\delta_t=(\delta_t^0,\delta_t^1,\ldots,\delta_t^d)^\top,\
%t\in [0,T]\}$.
%In general, we do not request strategies to be self-financing.
Denote by $\V_x^+$, ($\V_x$), the set of
all strictly positive, (nonnegative), finite, %benchmarked
self-financing portfolios, with initial
capital $x >0$, ($x \ge 0$).
\begin{definition} \label{def:numport}
A portfolio
$S^{\delta_*} \in \V_1^+$ is called a {\em num\'{e}raire portfolio}, if %it equals the constant one and if
any portfolio $S^\delta \in \V_1$, when denominated in units of $S^{\delta_*}$,
forms an $(\bF,\P)$-supermartingale, that is,
\begin{equation}\label{supermprop}
\frac{S_t^\delta}{S_t^{\delta_*}} \ge \condespf{\frac{S_s^\delta}{S_s^{\delta_*}}},
\end{equation}
for all $0 \leq t \leq s\leq T$.
\end{definition}
\noindent According to the benchmark approach we now make the main assumption of the paper.
\begin{ass}\label{ass:numport}
There exists a %benchmarked
num\'{e}raire portfolio $S^{\delta_*}
\in \V_1^+$.
\end{ass}
It is worth stressing that this assumption is rather weak since it is satisfied by a several realistic models used in practice, see e.g.~\cite{kk} and~\cite[Chapter 14, Section 1]{ph}.

\noindent Moreover, jump-diffusion and It\^o process driven market models have a num\'{e}raire portfolio under very general assumptions (see~\cite[Chapter 10, Section 2 and Chapter 14, Section 1]{ph})  where benchmarked nonnegative portfolios turn out to be $(\bF,\P)$-local martingales and, thus, $(\bF,\P)$-supermartingales.

\noindent Throughout the rest of the paper we choose the num\'{e}raire portfolio as {\em benchmark}. We call any security, when expressed in units of the num\'{e}raire portfolio, a benchmarked security and refer to this procedure as {\em benchmarking}. If a benchmarked price process is an $(\bF,\P)$-martingale, then we call it {\em fair}. In this case we would have equality in relation \eqref{supermprop} of Definition \ref{def:numport}. The benchmarked value of a portfolio $S^\delta$ is of particular interest and is given by the ratio
\[
\hat S_t^\delta=\frac{S_t^\delta}{S_t^{\delta_*}}
\]
for all $t \in [0,T]$. Of course the benchmarked value of the num\'{e}raire portfolio is $\hat S^{\delta^*}_t=1$ for every $t \in [0,T]$.
\noindent As usual in these situations, we want to exclude the presence of arbitrage at least in a strong form. The definition below provides the classical characterization of strong arbitrage.
\begin{definition}
A benchmarked nonnegative self-financing
portfolio $\hat S^\delta$ %\in \V_0$
is a {\em strong arbitrage} if
it starts with zero initial capital, that is $\hat S_0^\delta=0$, and
generates some strictly positive wealth with strictly positive
probability at some later time $t \in (0,T]$, that is
$\P(\hat S_t^\delta>0)>0$.
\end{definition}
%This means that, starting from zero capital at time $t=0$, it is not possible to generate strictly positive gains with strictly positive probability.
\noindent Note that, by the supermartingale property \eqref{supermprop}, assuming the existence of a strictly positive num\'{e}raire portfolio guarantees that strong arbitrage is
excluded, see e.g. Theorem 4.2 of~\cite{p2008}. Other classical forms of arbitrage have been investigated recently: for instance, the existence of the num\'{e}raire portfolio characterizes the absence of arbitrage in the sense of the NUPBR condition and the absence of arbitrage of the first kind, see e.g. in~\cite{kk} and~\cite{kar12}.

\noindent Let $\hat H_T$ be a nonnegative square integrable random variable representing the benchmarked payoff of a European type contingent claim with maturity $T$. It is known that in a complete market there always exists a
self-financing strategy $\delta$ that replicates this claim; in other words the associated portfolio value $\hat S_T^{\delta}$ at terminal time
equals $\hat H_T$ with probability one, see~\cite{ph}.
The benchmarked portfolio at every time $t\in [0,T]$ is then characterized by the real world pricing formula
\begin{equation} \label{realwpf}
\hat S_t^{\delta}=\condespf{\hat H_T}
\end{equation}
which means that the portfolio value is a square integrable $(\bF, \P)$-martingale, since $\hat H_T$ is square integrable.
As also observed in \cite{ph}, under the benchmark approach there exist other self-financing portfolios that deliver the claim at final time, but these nonnegative portfolios are just $(\bF,\P)$-supermartingales (and not $(\bF, \P)$-martingales), then they are more expensive than the $(\bF,\P)$-martingale $\hat S^{\delta}$ given in \eqref{realwpf}, see~\cite{ph} for further details.

\noindent When the market is incomplete it is impossible, in general, to find a self-financing portfolio that replicates the claim at final time. Then, one needs to relax one of these two conditions, the replication or the self-financing property, and provide a criterion to determine the optimal portfolio. The most common methods are the {\em mean-variance hedging} and the (local) {\em risk-minimization} which correspond to keep the self-financing property of the optimal strategy and the replication constraint respectively.
Here, we investigate the case of asset prices that represent $(\bF, \P)$-local martingales when benchmarked, and study the risk-minimizing approach in the case where there are restrictions on the available information.
When benchmarked primary security account processes are $(\bF,\P)$-local martingales, the problem can be solved under full information by computing the orthogonal projection of $\hat H_T$ on the set of stochastic integrals of the type $\int \gamma_u^\top \ud \hat S_u$, where $\gamma \in \Theta (\bF)$ 
(see Definition \ref{thetaH-F} below) and $\hat S=\{\hat S_t=(\hat S_t^0,\hat S_t^1,\ldots,\hat S_t^d)^\top,\ t \in [0,T]\}$ denotes the vector of the benchmarked primary security accounts, that is $\hat S_t^j = \frac{S_t^j}{S_t^{\delta_*}}$, $j=0,..,d$  (see~\cite{bcp} for further details).
The approach to the problem for partially observable models requires the application of a suitable orthogonal %version of the Galtchouk-Kunita-Watanabe
decomposition holding in the more general setting as we will see in the next section.

\section{Benchmarked Risk-Minimization under Restricted Information}\label{sec:brm}

Our aim is to discuss the risk-minimization as originally introduced in \cite{fs86}, under the benchmark approach in the case where there are restrictions on the available information.
%Risk-minimizing hedging under partial information using the savings account as reference unit has been studied in~\cite{s94} and in~\cite{ccr}.
Here, instead of the usual savings account, we use the num\'{e}raire portfolio as num\'{e}raire and benchmark and, in this setting, we wish to
construct risk-minimizing hedging strategies from less information. \\
Suppose now that the hedger does not have at her/his disposal the full information represented by $\bF$. Her/his strategy must be constructed from less information. To describe this mathematically, we introduce an additional filtration $\bH=\{\H_t, \ t\in [0,T]\}$ %satisfying the usual conditions
corresponding to the available information level to traders such that
\begin{equation}\label{filtr-cond}
\H_t \subseteq \F_t, \quad \mbox{for  every}\ t \in [0,T].
\end{equation}
%and $\H_T=\F_T=\F$.
We remark that all filtrations satisfy the usual conditions. Similarly to \cite{s94} we make the following assumption.
\begin{ass}\label{ass:misurabilitaS_T}
%As in \cite{s94},
We assume that at final time $T$, $\hat S_T$ is $\H_T$-measurable.
\end{ass}
We recall that under Assumption \ref{ass:numport}, the benchmarked
value of any nonnegative, self-financing portfolio forms an $\R^{d+1}$-valued $(\bF,\P)$-supermartingale, see \eqref{supermprop}. In
particular, the vector of the $d+1$ benchmarked primary security accounts
$\hat S$ forms with each of its components a nonnegative $(\bF,\P)$-{\it supermartingale}. By Theorem VII.12 of~\cite{dm2}, we know that the
vector process $\hat S$ has a unique decomposition of the form
\begin{equation} \label{supermartdeco}
\hat S_t=\hat S_0+M_t+A_t,\quad t \in [0,T],
\end{equation}
where $M=\{M_t=(M_t^0,M_t^1,\ldots,M_t^d)^\top,\ t \in [0,T]\}$ is an $\R^{d+1}$-valued $(\bF,\P)$-local martingale and $A=\{A_t=(A_t^0,A_t^1,\ldots,A_t^d)^\top,\ t \in [0,T]\}$ is an $\R^{d+1}$-valued $\bF$-predictable, in each component non-decreasing process with $M_0=A_0={\bf 0}$, with ${\bf 0}$ denoting the
$(d+1)$-dimensional null vector. This expresses the fact that every right-continuous $(\bF,\P)$-supermartingale is a special $(\bF,\P)$-semimartingale.\\
If the vector-valued process $\hat S=\{\hat S_t=(\hat S_t^0,\hat S_t^1,\ldots,\hat S_t^d)^\top,\ t \in [0,T]\}$ is continuous it is possible to provide an explicit description of the num\'{e}raire portfolio $S^{\delta_*}$ and the vector of primary security account processes
turns out to be an $\R^{d+1}$-valued $(\bF,\P)$-local martingale when expressed in units of the num\'{e}raire portfolio. On the other hand, if $\hat S$ exhibits jumps it is not possible to find an analogous characterization, see e.g.~\cite{bcp} for further details. However, a wide class of jump-diffusion market models is driven by primary security
account processes that are given by $(\bF,\P)$-local martingales, when expressed in units of the num\'{e}raire portfolio, as we will see for instance in Section \ref{sec:jumpdiff}. Here we discuss the hedging problem of a contingent claim in the partial information setting given by \eqref{filtr-cond} when benchmarked securities are $\R^{d+1}$-valued $(\bF,\P)$-local martingales. This justifies the following assumption.
\begin{ass}
The following holds:
$$
A_t \equiv 0, \quad \forall t \in [0,T]
$$
in \eqref{supermartdeco}.
\end{ass}
In order to deal with economically reasonable investment strategies,
we impose some integrability assumptions.
\begin{definition}\label{thetaH-F}
The space $\Theta(\bH)$ (respectively $\Theta(\bF)$) consists of all $\R^{d+1}$-valued $\bH$-predictable (respectively $\bF$-predictable) processes $\delta$ satisfying
the following integrability condition
\begin{equation}\label{admissible}
\esp{\int_0^T\delta_u^\top\cdot \ud \langle M\rangle_u \cdot \delta_u}<\infty.
\end{equation}
Here $\langle M\rangle=(\langle M^i,M^j\rangle)_{i,j=0,\ldots,d}$ denotes the $(d+1) \times (d+1)$ matrix-valued $\bF$-predictable sharp bracket process of the $\R^{d+1}$-valued $(\bF,\P)$-local martingale $M$.
\end{definition}
Since the sharp bracket of the semimartingale $\hat S$ is equal to the sharp bracket of its (local) martingale part, we use both the notations $\langle \hat S \rangle$ and $\langle M \rangle$.
\begin{definition}
An $\bH$-strategy is a pair $\phi=(\eta, \delta)$, where the process $\delta \in \Theta(\bH)$ describes the number of units invested in the benchmarked security accounts and $\eta=\{\eta_t, \ t \in [0,T]\}$ is an $\R$-valued $\bH$-adapted c\`{a}dl\`{a}g process such that the associated benchmarked portfolio value $\hat V^\phi=\{\hat V_t^\phi,\ t \in [0,T]\}$ is an $\bF$-adapted and  square-integrable process (i.e. $\hat V_t^\phi \in L^2(\F_t,\P)$, for each $t
\in [0,T]$) whose left limit is equal to $\hat
V_{t^-}^\phi=\sum_{j=0}^d \delta_t^j\hat S_{t^-}^j+ \eta_{t^-}$.
\end{definition}
\noindent The strategy component $\eta$ may be interpreted as the number of units invested in the num\'{e}raire portfolio $S^{\delta_*}$, see also Remark \ref{rem:3.8}.
\noindent Recall that the market may be not complete; then, we also admit strategies that are not self-financing and may generate (benchmarked) costs over time.
\begin{definition} For any $\bH$-strategy
$\phi$, the {\em benchmarked cost process} $\hat
C^\phi$ is defined by
\begin{equation}\label{profloss}
\hat C_t^\phi:=\hat V_t^{\phi}-\int_0^t \delta_u^\top \cdot \ud
\hat S_u, \quad \forall t \in [0,T].
\end{equation}
\end{definition}
\noindent Here $\hat C_t^\phi$ describes the total costs incurred over the time interval $[0,t]$.

In order to introduce the quadratic criterion of {\em risk-minimization} under restricted information we have to define the so-called {\em risk process} in such a framework.
\begin{definition} \label{def:risk}
For any $\bH$-strategy $\phi$, the corresponding $\bH$-{\em
risk} $\hat R^{\H,\phi}$ at time $t \in [0,T]$ is defined by
$$
\hat R_t^{\H,\phi}:=\condesph{\left(\hat C_T^\phi-\hat C_t^\phi
\right)^2}.
$$
\end{definition}
Note that the $\bH$-risk $\hat R^{\H,\phi}$ is well-defined, since
the benchmarked cost process $\hat C^\phi$, given
in \eqref{profloss}, is square-integrable. \\
Our goal is to find an $\bH$-strategy
$\phi$ which minimizes the associated $\bH$-risk measured by the
fluctuations of its benchmarked cost process in a suitable sense.

Throughout this paper, we consider a European style contingent claim with maturity $T$ whose benchmarked payoff
is given by an $\H_T$-measurable,
nonnegative random variable $\hat H_T$. We will always assume that a benchmarked contingent claim
$\hat H_T$ belongs to
$L^2(\H_T,\P)$.

For a benchmarked contingent claim $\hat H_T \in L^2(\H_T,\P)$, it makes sense to define some risk-minimizing strategy under restricted information
by
looking for an $\bH$-strategy $\phi=(\eta,\delta)$ such that $\hat V_T^\phi =\hat H_T$ which minimizes the $\bH$-risk process $\hat R^{\H,\phi}$ in the following manner:
\begin{definition} \label{optimalstrategy}
Given a benchmarked contingent claim $\hat H_T \in L^2(\H_T,\P)$,
an  $\bH$-strategy $\phi=(\eta, \delta)$ is said to be {\em
benchmarked %locally
$\bH$-risk-minimizing}
if the following conditions hold:
\begin{itemize}
\item[(i)] $\hat V^\phi_T=\hat H_T,\ \P$-\mbox{a.s.};
\item[(ii)] for any $t \in [0,T]$ and for any $\bH$-strategy $\tilde \phi$ such that $\hat V_T^{\tilde \phi}=\hat V_T^\phi$ $\P$-\mbox{a.s.}, then
$$
\hat R_t^{\H,\phi} \leq \hat R_t^{\H,\tilde \phi},\quad \P-\mbox{a.s.}.
$$
\end{itemize}
\end{definition}

\begin{remark}\label{rem:3.8}
We observe that for an $\bH$-strategy $\phi=(\eta, \delta)$, at any time $t \in [0,T]$ the benchmarked portfolio value is given by
\[
\hat V^\phi_t= \sum_{j=0}^d \delta^j_t \hat S^j_t + \eta_t =  \sum_{j=0}^d \bar \delta^j_t \hat S^j_t =:\hat S^{\bar \delta}_t
\]
where $\bar \delta^j_t= \delta^j_t + \eta_t \delta^j_{*,t}$ and $\hat S^{\bar \delta}=\{\hat S_t^{\bar \delta},\ t \in [0,T]
\}$ is the benchmarked portfolio value associated to the strategy $\bar \delta:=\{\bar \delta_t=(\bar \delta_t^0,\bar \delta_t^1,\ldots,\bar \delta_t^d)^\top, \ t \in [0,T]\}$. We recall that
\[
\sum_{j=0}^d \delta^j_{*,t} \hat S^j_t= \hat S^{\delta_{*}}_t=1, \quad t \in [0,T],
\]
and $\delta^j_{*,t}$ denotes the number of units of the $j$-th benchmarked security held at time $t$ in the num\'{e}raire portfolio.

The strategies $\phi$ and $\bar \delta$ have the same benchmarked cost process. In fact, for every $t \in [0,T]$, we have
\begin{gather*}
\hat C_t^\phi = \hat V_t^{\phi}-\int_0^t \delta_u^\top \cdot \ud \hat S_u = \hat S_t^{\bar \delta}-\int_0^t \delta_u^\top \cdot \ud \hat S_u \\
=\hat S_t^{\bar \delta}-\int_0^t (\bar \delta_u)^\top \cdot \ud \hat S_u + \int_0^t \eta_u (\delta_{*,u})^\top \cdot \ud \hat S_u  \\
=\hat S_t^{\bar \delta}-\int_0^t (\bar \delta_u)^\top \cdot \ud \hat S_u + \int_0^t \eta_u \ud \hat S^{\delta_{*}}_u\\
=\hat S_t^{\bar \delta}-\int_0^t (\bar \delta_u)^\top \cdot \ud \hat S_u = \hat C_t^{\bar \delta}.
\end{gather*}

\medskip

We also note that the $\R^{d+1}$-valued process $\bar \delta$ is in general $\bF$-adapted, since $\eta$ is $\bH$-adapted and $\delta^{j}_{*}$ is $\bF$-predictable for every $j=0,...,d$.
If $\phi=(\eta, \delta)$ is any benchmarked $\bH$-risk-minimizing strategy and $\delta_*$ is $\bH$-predictable, then $\bar \delta$ is an $\bH$-adapted benchmarked risk-minimizing strategy that only requires to invest in the benchmarked assets $\hat S^0$,...,$\hat S^d$.
\end{remark}

\subsection{The variational formulation of the benchmarked $\bH$-risk-minimizing problem }

To introduce some useful notations, we define $\mathcal M_0^2(\bF,\P)$
as the space of all square-integrable $\R$-valued $(\bF,\P)$-martingales null at time $t=0$.

Let $\hat H_T \in L^2(\H_T,\P)$ and consider the well-known Galtchouk-Kunita-Watanabe decomposition of $\hat H_T$ with respect to $\hat S$ under full information:
\begin{equation} \label{GKWdecomp}
\hat H_T= \hat H_0+\int_0^T (\delta^\F_u)^\top \cdot \ud \hat
S_u+\tilde L_T^{\hat H},\quad \P-{\rm a.s.},
\end{equation}
where $\hat H_0 \in L^2(\F_0,\P) $, $ \delta^\F\in \Theta(\bF)$  and $\tilde L^{\hat H}=\{\tilde L_t^{\hat H}, \ t \in [0,T]\} \in \mathcal M_0^2(\bF,\P)$ is such that $\langle \tilde L^{\hat H}, \hat S^j\rangle_t = 0$, for every $t \in [0,T]$ and $j=0,1,\ldots,d$.

Decomposition \eqref{GKWdecomp} is an essential tool to rewrite the risk-minimizing problem in variational terms. Moreover, the variational formulation will allow us to derive an important feature of the benchmarked cost process associated to the optimal strategy.

We start with a useful lemma which gives us the martingale property of the cost process.

\begin{lemma}
For any $\bH$-strategy $\phi=(\eta, \delta)$ and any $t_0 \in [0,T]$, there exists an $\bH$-strategy $\tilde \phi= (\tilde \eta, \delta)$ such that
\begin{gather*}
\hat V^\phi_T= \hat V^{\tilde \phi}_T \quad \P-a.s.\\
\eta_t=\tilde \eta_t \quad \mbox{ for every }  t<t_0\\
\condesph{\hat C_T^{\tilde \phi}- \hat C_t^{\tilde \phi}}=0     \quad \P-a.s.  \mbox{ for every }  t \geq t_0\\
\end{gather*}
and
\[
R^{\H, \tilde \phi}_t \leq R^{\H, \phi}_t \quad \P-a.s. \mbox{ for every } t \geq t_0
\]
and $\tilde \eta$ can be chosen to satisfy
\[
\tilde \eta_t = \condesph{\hat H_T- \delta^\top_t \hat S_t} \quad \mbox{ for every } t \geq t_0.
\]
\end{lemma}

The proof follows the same lines of the proof of Lemma 2.1 in \cite{s94} with the choice $\mathcal{G}_t= \mathcal{G}_t'=\H_t$ for every $t\in [0,T]$. As a consequence, choosing $t_0=0$, we get that the $\bH$-optional projection of the benchmarked cost process $\hat C^{\tilde \phi}$ is an $\bH$-martingale.

We now introduce the variational formulation of the risk-minimizing problem.

\begin{lemma}\label{lemma:optim}
Let $\hat H_T \in L^2(\H_T,\P)$ be a benchmarked contingent claim and $\phi=(\eta, \delta)$ an $\bH$-strategy that delivers the claim at time $T$ (i.e. $\hat V^{\phi}_T=\hat H_T$). Then,  $\phi$ is benchmarked $\bH$-risk-minimizing if and only if the $\bH$-optional projection of the associated benchmarked cost process $\hat C^{\phi}$ is an  $\bH$-martingale and solves the following optimization problem:
\begin{equation}\label{optim.problem}
\min_{\gamma \in \Theta (\bH)} \esp{\left(\hat H_T - \int_0^T\gamma_u^\top \cdot \ud \hat S_u\right)^2}.
\end{equation}
\end{lemma}

The proof of this lemma follows by Proposition 2.3 in \cite{s94}.

\begin{remark}\label{rem:orthogonality}
Note that, given a benchmarked $\bH$-risk-minimizing strategy $\phi = (\eta, \delta)$, the associated residual benchmarked cost process $\hat C_T^\phi- \hat C_t^\phi$ at time $t \in [0,T]$ satisfies the following condition
%is weakly orthogonal to all the integrals of the form $\ds \int_t^T \gamma_u \ud \hat S_u$, with $\gamma \in \Theta(\bH)$, thanks to Definitions \ref{def:risk} and \ref{optimalstrategy}, that means
$$
\condesph{\left(\hat C_T^\phi- \hat C_t^\phi\right)\int_t^T \gamma_u^\top \cdot \ud \hat S_u}=0, \quad \P-a.s.
$$
for every $t \in [0,T]$ and for each $\gamma \in \Theta(\bH)$.

Indeed, let $\hat H_T \in L^2(\H_T,\P)$ be a benchmarked contingent claim. If an $\bH$-strategy $\phi = (\eta, \delta)$ is benchmarked $\bH$-risk-minimizing, then by Lemma \ref{lemma:optim}, we have that the second component of $\phi$ solves the following optimization problem:
\begin{equation*}
\min_{\gamma \in \Theta (\bH)} \esp{\left(\hat H_T - \int_0^T\gamma_u^\top \cdot\ud \hat S_u\right)^2}.
\end{equation*}
In particular, thanks to the projection theorem, this is equivalent to
\begin{gather*}\label{eq:proj_thm}
\esp{\left(\hat V_T^\phi - \int_0^T\delta_u^\top \cdot \ud \hat S_u\right) \int_0^T\gamma_u^\top \cdot \ud \hat S_u}  = 0,
\end{gather*}
for every $\gamma \in \Theta (\bH)$
and by \eqref{profloss} also to
\begin{equation*}\label{eq:ortogonalitaC}
\esp{\hat C_T^\phi \int_0^T\gamma_u^\top \cdot \ud \hat S_u}  = 0
\end{equation*}
for every $\gamma \in \Theta (\bH)$.
Finally, by Lemma 5.4 of~\cite{ccr}, this is also equivalent to
$$
\condesph{\left(\hat C_T^\phi- \hat C_t^\phi\right)\int_t^T \gamma_u^\top \cdot \ud \hat S_u}=0, \quad \P-a.s.
$$
for every $t \in [0,T]$ and every $\gamma \in \Theta(\bH)$.
 Note that we could apply
Lemma 5.4 of~\cite{ccr} even if we do not have the martingale property of the process $\hat C^\phi$.
\end{remark}

\subsection{Existence of the benchmarked $\bH$-risk-minimizing strategy}

In this section we wish to characterize explicitly the benchmarked $\bH$-risk minimizing strategy %$\delta^\H$ given in \eqref{eq:GKWpartial}
in terms of the integrand $\delta^\F$ appearing in decomposition \eqref{GKWdecomp}.

We start with a generalization of Proposition 4.3 in~\cite{ccr}.

Let $G=\{G_t=(G_t^0,G_t^1,\ldots,G_t^d), \ t\in [0,T]\}$ be an $\R^{d+1}$-valued $\bF$-adapted c\`{a}dl\`{a}g process of integrable variation. We denote by $\|G\|=\{\|G\|_t=(\|G^0\|_t,\ldots,\|G^{d}\|_t),\ t \in [0,T]\}$ the total variation of the function $t \to G_t(\omega)$ defined by
\[
\|G^j\|_t(\omega)=\sup_{\Delta}\sum_{i=0}^{n(\Delta)-1}|G^j_{t_i+1}(\omega)-G^j_{t_i}(\omega)|, \quad j \in \{0,1,\ldots,d\}
\]
where $\Delta:=\{t_0=0<t_1<...<t_n=t\}$ is a partition of the interval $[0,t]$.

\begin{proposition}\label{prop:dualproj}
Let $G$ be an $\R^{d+1}$-valued c\`{a}dl\`{a}g $\bF$-adapted  process of integrable variation. Then, there exists a unique $\R^{d+1}$-valued $\bH$-predictable process $G^\bH=\{G_t^\bH,\ t\in [0,T] \}$ of integrable variation such that
\[
\esp{\int_0^T \theta_t^\top \cdot \ud G_t }=\esp{\int_0^T \theta_t^\top \cdot \ud G_t^\bH }
\]
for every $\R^{d+1}$-valued $\bH$-predictable (bounded) process $\theta$. The process $G^\bH$ is called the $\bH$-predictable dual projection of $G$.
\end{proposition}

\begin{proof}
First, we can observe that
\begin{equation}\label{eq:esp1}
\esp{\int_0^T \theta_t^\top \cdot \ud G_t }=\esp{\sum_{j=0}^{d}\int_0^T \theta^j_t \ud G^j_t}=\sum_{j=0}^{d}\esp{\int_0^T \theta^j_t \ud G^j_t}.
\end{equation}
Then, for every $j=0,...,d$ we can apply Proposition 4.3 in \cite{ccr} so that \eqref{eq:esp1} becomes
\[
\sum_{j=0}^{d}\esp{\int_0^T \theta^j_t \ud G^j_t}=\sum_{j=0}^{d}\esp{\int_0^T \theta^j_t \ud (G^{j}_t)^\bH}=\esp{\sum_{j=0}^{d}\int_0^T \theta^j_t \ud (G^{j}_t)^\bH}=\esp{\int_0^T \theta_t^\top \cdot \ud G^\bH_t },
\]
which yields the result.
\end{proof}

Similarly to~\cite{s94}, we fix an increasing $\bF$-predictable c\`{a}dl\`{a}g process $B$ null at initial time such that $\langle \hat S^i, \hat S^j \rangle=\langle M^i, M^j\rangle$ is absolutely continuous with respect to $B $ for every $i,j=0,...,d$ and we write $\langle \hat S^i, \hat S^j \rangle<<B$.
Then, we define the $\bF$-predictable $(d+1) \times (d+1)$ matrix-valued process $\sigma$ by setting, for every $t \in [0,T]$
\begin{equation} \label{eq:sigma}
\sigma^{i,j}_t:=\frac{\ud \langle \hat S^i, \hat S^j\rangle_t}{\ud B_t}\quad \mbox{for} \ i,j=0,1,\ldots,d.
\end{equation}
Then, the space $\Theta(\bH)$ can be rewritten as the set of all $\R^{d+1}$-valued $\bH$-predictable processes $\delta$ such that
\[
\esp{\int_0^T\delta_t^\top \cdot \sigma_t \cdot \delta_t \ \ud B_t}<\infty.
\]
\begin{remark}\label{rem:abs_continuity}
Let $Z$ and $\tilde Z$ be locally integrable c\`{a}dl\`{a}g processes of finite variation. If $Z<<\tilde Z$ then also $Z^\bH<<\tilde Z^\bH$.
\end{remark}
Hence, we can define the $\bH$-predictable $(d+1) \times (d+1)$ matrix-valued process $\varrho$ by setting for each $t \in [0,T]$
\begin{equation*}\label{eq:rho}
\varrho_t^{i,j}:=\frac{\ud \left(\int \sigma^{i,j}_u \ud B_u\right)_t^\bH}{\ud B^\bH_t} \quad \mbox{for} \ i,j=0,1,\ldots,d.
\end{equation*}
Indeed, we can easily check that
since $\langle \hat S^i, \hat S^j\rangle << B$  for $i,j=0,1,\ldots,d$, then $\langle \hat S^i, \hat S^j\rangle^\bH<< B^\bH$ for $i=0,1,\ldots,d$ (Remark \ref{rem:abs_continuity}). Therefore, for each $t \in [0,T]$ we can set
\[
\varrho^{i,j}_t=\frac{\ud \langle \hat S^i, \hat S^j\rangle^\bH_t}{\ud B^\bH_t}\quad \mbox{for} \ i,j=0,1,\ldots,d.
\]
Finally, note that by \eqref{eq:sigma} we have that
$\langle \hat S^i, \hat S^j\rangle^\bH=\left(\int \sigma_u^{i,j}\ud B_u\right)^\bH$ for $i,j=0,1,\ldots,d$.

\begin{lemma}\label{lemma:uniq_var_problem}
For every benchmarked contingent claim $\hat H_T \in L^2(\H_T,\P)$ the variational problem admits a unique solution $\delta\in \Theta(\bH)$.
\end{lemma}

\begin{proof}
The proof follows the same lines of that of Lemma 2.4 in \cite{s94}. Since the integral $\int_0^T \gamma_u^\top \cdot \ud \hat S_u$
is an isometry from $\Theta(\bF)$ to $L^2(\H_T , \P)$, it is sufficient to prove that $\Theta(\bH)$ is a closed subspace of $\Theta(\bF)$.
In fact by \eqref{eq:rho}, for every $\gamma \in \Theta (\bH)$ we get
\begin{align*}
&\esp{\int_0^T\gamma_u^\top \cdot \sigma_u \cdot \gamma_u \ud B_u}=\sum_{i,j=0}^d\esp{\int_0^T \gamma_u^i \gamma_u^j \ud \left(\int \sigma^{i,j}\ud B\right)_u}=\sum_{i,j=0}^d\esp{\int_0^T \gamma_u^i \gamma_u^j \ud \left(\int \sigma^{i,j}\ud B\right)^\bH_u}\\
&=\sum_{i,j=0}^d\esp{\int_0^T \gamma_u^i \gamma_u^j \ud \left(\int \varrho^{i,j}\ud B^\bH\right)_u}=\esp{\int_0^T\gamma_u^\top \cdot \varrho_u \cdot \gamma_u \ud B^\bH_u}.
\end{align*}
Since $\varrho$ and $B^\bH$ are both $\bH$-predictable then we get that $\Theta(\bH)$ is closed in $\Theta(\bF)$.
\end{proof}

The following result extends Theorem 2.5 in~\cite{s94}
to the benchmarked framework and provides an explicit representation for the benchmarked $\bH$-risk-minimizing strategy in terms of the integrand appearing in the classical Galtchouk-Kunita-Watanabe decomposition.
\begin{proposition} \label{explicit_delta}
For any $\hat H_T \in L^2(\H_T, \P)$ there exists a unique benchmarked $\bH$-risk-minimizing strategy $\phi^\H=(\eta^\H, \delta^\H)$ that is given by
\begin{gather}
\delta_t^\H=\varrho^{-1}_t\cdot \frac{\ud \left(\int \sigma_u \cdot \delta_u^\F \ud B_u\right)^\bH_t}{\ud B^\bH_t}\label{explicit_deltaH},\\
\eta_t^\H=\condesph{\hat H_T-(\delta^\H_t)^\top \cdot \hat S_t}\nonumber
\end{gather}
for every $t \in [0,T]$, where $\varrho^{-1}$ is the pseudo-inverse of the matrix valued process $\varrho$ and $\delta^\F$ is the process given in \eqref{GKWdecomp}.
\end{proposition}

\begin{proof}
Existence and uniqueness of $\phi^\H$ are ensured by Lemma \ref{lemma:optim} and Lemma \ref{lemma:uniq_var_problem}.
To compute $\delta^\H$ we observe that by the projection theorem a process $\delta \in \Theta(\bH)$ solves the optimization problem \eqref{optim.problem} if and only if
\begin{equation}\label{eq:proj}
\esp{\left(\hat H_T - \int_0^T\delta_u^\top \cdot \ud \hat S_u\right)\int_0^T\gamma_u^\top \cdot \ud \hat S_u}  = 0, \quad {\rm for\ every}\ \gamma \in \Theta (\bH).
\end{equation}
Since every benchmarked claim $\hat H_T$ admits the Galtchouk-Kunita-Watanabe decomposition with respect to $\hat S$ under complete information, see  \eqref{GKWdecomp}, equation \eqref{eq:proj} can be rewritten as
\begin{equation}\label{eq:proj2}
\esp{\int_0^T\left(\delta^\F_u-\delta_u\right)^\top \cdot \ud \hat S_u \int_0^T\gamma_u^\top \cdot \ud \hat S_u}  = 0, \quad \mbox{for\ every}\ \gamma \in \Theta (\bH).
\end{equation}
Then the strategy $\delta^\H$ is determined by condition
$$
0=\esp{\left(\hat H_T - \int_0^T (\delta^\H_u)^\top \cdot \ud \hat S_u\right)\int_0^T\gamma_u^\top \cdot \ud \hat S_u}  =  \esp{\left(\int_0^T(\delta^\F_u - \delta^\H_u)^\top \cdot \ud \hat S_u\right)\int_0^T \gamma_u^\top \cdot \ud \hat S_u}
$$
for\ every $\gamma \in \Theta (\bH)$,
which means that $\delta^\H$ has to satisfy
\begin{equation}\label{eq:tecnica}
\esp{\int_0^T(\delta^\F_u)^\top \cdot \ud \hat S_u\int_0^T \gamma_u^\top \cdot \ud \hat S_u}=\esp{\int_0^T(\delta^\H_u)^\top \cdot \ud \hat S_u \int_0^T \gamma_u^\top \cdot \ud \hat S_u}
\end{equation}
for\ every $\gamma \in \Theta (\bH)$. Note that the expectation on the left-hand side can be rewritten as
\begin{align*}
&\esp{\int_0^T (\delta^\F_u)^\top \cdot \ud \hat S_u \int_0^T \gamma_u^\top \cdot \ud \hat S_u}=\esp{\int_0^T  \gamma_u^\top \cdot \sigma_u \cdot \delta^\F_u \ud B_u}\\
&\quad =\esp{\int_0^T  \gamma_u^\top \cdot \ud \left(\int \sigma_r \cdot \delta_r^\F \ud B_r \right)_u}=\esp{\int_0^T  \gamma_u^\top \cdot \ud \left(\int \sigma_r \cdot \delta_r^\F \ud B_r \right)^\bH_u},
\end{align*}
where the last equality follows by Proposition \ref{prop:dualproj}. Moreover, the term on the right-hand side is given by
\begin{align*}
&\esp{\int_0^T(\delta^\H_u)^\top \cdot \ud \hat S_u \int_0^T \gamma_u^\top \cdot \ud \hat S_u}=\esp{\int_0^T   \gamma_u^\top \cdot \sigma_u \cdot \delta^\H_u \ud B_u}\\
&\quad =\sum_{i,j=0}^d \esp{\int_0^T  \gamma_u^i (\delta_u^\H)^j \ \sigma^{i,j}_u \ \ud B_u }=\sum_{i,j=0}^d \esp{\int_0^T  \gamma_u^i  (\delta_u^\H)^j \ud \left(\int \sigma_r^{i,j} \ud B_r \right)_u}\\
&\quad =\sum_{i,j=0}^d \esp{\int_0^T  \gamma_u^i (\delta_u^\H)^j  \ud \left(\int \sigma_r^{i,j} \ud B_r \right)^\bH_u}=\sum_{i,j=0}^d \esp{\int_0^T  \gamma_u^i (\delta_u^\H)^j \ \varrho^{i,j}_u \ \ud B^\bH_u }\\
&\quad = \esp{\int_0^T   \gamma_u^\top \cdot \varrho_u \cdot \delta^\H_u \ \ud B^\bH_u},
\end{align*}
where we used again Proposition \ref{prop:dualproj} and the definition of the $(d+1) \times (d+1)$ matrix-valued process $\varrho$.
Then, equality \eqref{eq:tecnica} becomes
\[
\esp{\int_0^T  \gamma_u^\top \cdot \ud \left(\int \sigma_r \cdot \delta_r^\F \ud B_r \right)^\bH_u}= \esp{\int_0^T  \gamma_u^\top \cdot \varrho_u \cdot \delta^\H_u \ \ud B^\bH_u}
\]
for\ every $\gamma \in \Theta (\bH)$, which leads to
\[
\varrho_t \cdot \delta^\H_t =\frac{\ud \left(\int \sigma_r \cdot \delta^\F_r \ud B_r \right)^\bH_t}{\ud B^\bH_t}, \quad \mbox{for every }\ t \in [0,T]
\]
and finally to
\[
\delta^\H_t =\varrho_t^{-1}\cdot \frac{\ud \left(\int \sigma_r \cdot \delta_r^\F \ud B_r \right)^\bH_t}{\ud B^\bH_t},\quad \forall t \in [0,T].
\]
\end{proof}

The proposition above provides a practical method to characterize the process $\delta^\H $ in terms of $\delta^\F$.
We now show an example.

\begin{example}\label{esempio1}
Suppose that each component of $\langle \hat S\rangle$ is absolutely continuous with respect to the Lebesgue measure, i.e.
\[
 \ud \langle \hat S^i, \hat S^j \rangle_t =a^{i,j}(t) \ud t, \quad \mbox{for}\ i,j=0,1,\ldots,d
\]
for each $t \in [0,T]$. Here $a^{i,j}=\{a^{i,j}(t),\ t\in [0,T]\}$ are $\bF$-predictable processes representing the $(d+1) \times (d+1)$ components of some matrix-valued process $a$. %and $G^{i,j}$ increasing deterministic functions.
According to \eqref{eq:sigma}, we can choose $B_t = t$, for every $t \in [0,T]$ and then the $(d+1) \times (d+1)$ matrix-valued process $\sigma$ coincides with $a$ and in particular $\varrho$ is given by
\[
\varrho^{i,j}_t={}^p\left( a^{i,j}(t) \right), \quad \mbox{for}\ i,j=0,1,\ldots,d
%\quad \mbox{and}\quad
%A^{\bH,j}_t=\sum_{i=0}^{d}\int_0^t  {}^p(\delta^{\F,i}_s a^{i,j}(s)) \ud s.
\]
for each $t \in[0,T]$, where the notation ${}^pZ$ represents the $\bH$-predictable projection of an integrable process $Z$, i.e. ${}^pZ_t:=\esp{Z_t|\H_{t^-}}$ for every $t \in [0,T]$. Assume that the the $(d+1) \times (d+1)$ matrix-valued process $\varrho_t = {}^p a(t)$ is invertible. Then, thanks to \eqref{explicit_deltaH}, the strategy component $\delta^\H$ is given by
\begin{equation}\label{eq:deltah}
\delta^\H_t=\varrho^{-1}_t \cdot {}^p\left(\sigma_t \cdot \delta_t^\F \right)= ({}^p a(t))^{-1} \cdot \ {}^p\left(a(t) \cdot \delta_t^\F\right), \quad \forall t \in [0,T].
\end{equation}

\end{example}

\subsection{The $\bH$-benchmarked risk-minimizing strategy and the Galtchouk-Kunita-Watanabe decomposition}

In this section we will see that finding a benchmarked $\bH$-risk-minimizing strategy corresponds
to looking for a suitable decomposition of the benchmarked claim that works under restricted information. We start with the following definition of orthogonality.
\begin{definition}
Let $X$ be an $\R$-valued $(\bF,\P)$-martingale.
We say that $X$ is {\em weakly orthogonal} to
$\hat S$ if
%all the integrals of the form $\ds \int_0^T \gamma_u \cdot \ud Y_u$, with $\gamma \in \Theta(\bH)$, if
\begin{equation}\label{orthogonality}
\esp{ X_T\int_0^T \gamma_u^\top \cdot \ud  \hat S_u}=0
\end{equation}
for every $\ds \gamma \in \Theta(\bH)$.
\end{definition}

For reader's convenience, we briefly discuss in the following remark the relationship between the orthogonality condition \eqref{orthogonality} and the (classical) strong orthogonality condition for $(\bF,\P)$-martingales.

\begin{remark}
Let $X$ be an $\R$-valued $(\bF,\P)$-martingale. Since for any $\bH$-predictable process $\gamma$, the process
$\I_{(0,t]}(s) \gamma_s$, with $t \leq T$,
is still $\bH$-predictable,  condition \eqref{orthogonality}
implies that for every $t \in [0,T]$
$$
\esp{X_T \int_0^t \gamma_u^\top \cdot \ud \hat S_u}=0,
$$
and
by conditioning with respect to $\F_t$ (note that $X$ is an $(\bF,\P)$-martingale), we have
$$
\esp{X_t \int_0^t \gamma_u^\top \cdot \ud \hat S_u}= \esp{\int_0^t \gamma_u^\top \cdot \ud \langle X,\hat S \rangle_u}= 0 \quad \mbox{for\ every}\ t \in [0,T], \   \mbox{and}  \   \gamma \in \Theta(\bH).
$$
From this last equality, we can argue that in the case of full information, i.e., $\H_t=\F_t$, for each $t \in [0,T]$, condition \eqref{orthogonality} is equivalent to the strong orthogonality condition between $X$ and $\hat S$ (see e.g. Lemma 2 and Theorem 36, Chapter IV, page 180 of~\cite{pp2004} for a rigorous proof).
\end{remark}

The following result provides a martingale representation of the benchmarked contingent claim that separates its hedgeable part from its non-hedgeable part (see~\cite[Chapter 11, Section 5]{ph} for a further discussion on this issue).

\begin{proposition}\label{prop:GKWpartial}
Let $\hat H_T\in L^2(\H_T,\P)$. Then, the random variable $\hat H_T$ can be uniquely written as
\begin{equation}\label{eq:GKWpartial}
\hat H_T= \hat H_0 + \int_0^T  (\xi^\H_u)^\top \cdot \ud \hat S_u + L^{\hat H}_T,\quad \P-\mbox{a.s.},
\end{equation}
where $\hat H_0\in L^2(\F_0, \P)$, $\xi^\H\in \Theta(\bH)$ and $L^{\hat H}=\{L^{\hat H}_t, \ t\in [0,T]\} \in \mathcal M_0^2(\bF,\P)$ is weakly orthogonal to $\hat S$. Moreover $\hat H_0=\condespho{\hat H_T}$.
\end{proposition}
We refer to \eqref{eq:GKWpartial} as the Galtchouk-Kunita-Watanabe decomposition of $\hat H_T$ with respect to $\hat S$ under partial information.

\begin{proof}
This result extends Proposition 3.2 in \cite{ccr} to the multidimensional case.
The proof uses the same techniques since it is not necessary that $\hat S$ is a true $(\bF, \P)$-martingale. Indeed, it is sufficient that the process $\int (\xi^\H_u)^\top \cdot \ud \hat S_u$ is a true $(\bF, \P)$-martingale, but this property is implied by condition \eqref{admissible}.
\end{proof}

In the sequel we wish to discuss the relationship between the component $\delta^\H$ of the $\bH$-benchmarked risk minimizing strategy $\phi^\H$ and the Galtchouk-Kunita-Watanabe decomposition of the claim $\hat H_T$ under partial information given by \eqref{eq:GKWpartial}.

\begin{theorem} \label{prop:fs}
The benchmarked %locally
$\bH$-risk-minimizing strategy $\phi^\H=(\eta^\H, \delta^\H)$ is uniquely characterized by taking $\delta^\H$ equal to the integrand of the Galtchouk-Kunita-Watanabe decomposition of $\hat H_T$ with respect to $\hat S$ under partial information, see \eqref{eq:GKWpartial}, i.e.
$$
\delta^\H =\xi^\H.
$$
In particular, the
minimal benchmarked cost and the optimal benchmarked portfolio value satisfy respectively
\begin{equation}\label{cost_optional_proj}
\condesph{ \hat C^{\phi^\H}_t} = \condesph{\hat H_0} + \condesph{L^{\hat H}_t}
\end{equation}
and
\begin{equation}\label{eq:portfolio_proj}
 \condesph{\hat V_t^{\phi^\H}} = \condesph{\hat H_T}
 \end{equation}
for every $t \in [0,T]$.
\end{theorem}

\begin{proof}
 Set $\tilde \phi^\H:=(\eta^\H, \xi^\H)$. To prove the first part of the statement it is sufficient to show that the portfolio value  $\hat V^{\tilde \phi^\H}$ at time $T$ replicates $\hat H_T$ and $\xi^\H$ solves the optimization problem in Lemma \ref{lemma:optim}, or equivalently that $\hat V_T^{\tilde \phi^\H} = \hat H_T$ $\P-a.s.$  and by \eqref{eq:proj2} that 

\begin{equation}\label{eq:proj3}
\esp{\int_0^T\left(\delta^\F_u -\xi^\H_u\right)^\top \cdot \ud \hat S_u \int_0^T\gamma_u^\top \cdot \ud \hat S_u}  = 0, \quad \mbox{for\ every}\ \gamma \in \Theta (\bH).
\end{equation}
By decompositions   \eqref{GKWdecomp} and \eqref{eq:GKWpartial}, when conditioning to $\F_t$ we get

$$ \int_0^t  (\xi^\H_u)^\top \cdot \ud \hat S_u + L^{\hat H}_t=  \int_0^t (\delta^\F_u)^\top \cdot \ud \hat S_u +  \tilde L^{\hat H}_t.$$

Hence the process
$$\int_0^t  (\xi^\H_u-\delta^\F_u)^\top \cdot \ud \hat S_u = \tilde  L^{\hat H}_t -  L^{\hat H}_t$$

turns to be an $(\bF, \P)$-martingale weakly orthogonal to $\hat S$ and this proves \eqref{eq:proj3}.% and concludes the first part of the proof.

Finally to prove \eqref{cost_optional_proj} and \eqref{eq:portfolio_proj} we observe that the process
$\condesph{\hat V^{\phi^\H}_t}$ for every $t \in [0,T]$ is an $(\bH,\P)$-martingale because $\condesph{\hat C^{\phi^\H}_t}$ and $\condesph{\int_0^t (\delta^\H_u)^\top \cdot \ud \hat S_u}$ for every $t \in [0,T]$ are $(\bH,\P)$-martingales. Then, for every $t \in [0,T]$,
\[
\condesph{\hat V^{\phi^\H}_t}=\condesph{\hat V^{\phi^\H}_T}=\condesph{\hat H_T}.
\]
As a consequence the benchmarked cost process satisfies
\begin{align*}
&\condesph{ \hat C^{\phi^\H}_t} = \condesph{\hat V^{\phi^\H}_t-\int_0^t(\delta^\H_u)^\top \cdot \ud \hat S_u}= \condesph{\hat H_T - \int_0^t(\delta^\H_u)^\top \cdot \ud \hat S_u} \\
& = \condesph{\hat H_0 + \int_0^T (\delta^\H_u)^\top \cdot \ud \hat S_u + L^{\hat H}_T- \int_0^t (\delta^\H_u)^\top \cdot \ud \hat S_u} = \condesph{\hat H_0} + \condesph{L^{\hat H}_t},
\end{align*}
for every $t \in [0,T]$.
\end{proof}

\begin{remark}
If the filtration $\bH$ corresponds to the filtration generated by the benchmarked prices, that is, $\bH=\bF^{\hat S}:=\{\F^{\hat S}_t, \ t \in [0,T]\}$ with $\F^{\hat S}_t :=\{\hat S_u, \ 0 \leq u \leq t\}$, then we get the following expressions for the benchmarked portfolio value and the benchmarked cost processes respectively:
\begin{gather*}
\hat V^{\phi^\H}_t=\bE\left[\hat H_T |\F^{\hat S}_t\right]\\
\hat C^{\phi^\H}_t=\bE\left[\hat H_0 |\F^{\hat S}_t\right] + \bE\left[\hat L^{\hat H}_T |\F^{\hat S}_t\right]
\end{gather*}
for every $t \in [0,T]$, instead of \eqref{cost_optional_proj}
and \eqref{eq:portfolio_proj}. This is implied by the fact that $\hat V^{\phi^\H}$ and $\hat C^{\phi^\H}$ are $\bF^{\hat S}$-adapted and then, by the above equalities, $(\bF^{\hat S}, \P)$-martingales. %In this case we get an explicit characterization of the portfolio value and the cost of optimal strategy $\phi^\H$ instead of the analogous characterization of the $\bH$-projections.
\end{remark}

\section{An $\bH$-benchmarked risk-minimizing strategy for a Markovian jump-diffusion market model}\label{sec:jumpdiff}

We now apply the benchmarked risk-minimization under partial information to a financial market affected by the presence of jumps in the underlying primary security account processes.
%We consider a financial market where two assets $(S^0,S^1)$ are traded continuously up to some fixed time horizon $T \in (0,\infty)$.
\subsection{The Markovian jump-diffusion market model}
For the sake of simplicity we restrict the market model introduced in Section \ref{sec:numeraire} to the case of a primary security account process $S^1=\{S_t^1, \ t \in [0,T]\}$ and a savings account $S^0=\{S^0_t, \ t \in [0,T]\}$ living on the probability space $(\Omega,\F,\P)$ endowed with the filtration $\bF$ . Recall that $T \in (0,\infty)$ denotes the finite time horizon.
On this probability space we define two standard one-dimensional $\bF$-Brownian motions $U=\{U_t, \ t \in [0,T]\}$ and $W=\{W_t,\ t \in [0,T]\}$ such that $\langle U, W\rangle_{t}= \rho \ t$ for every $t \in [0,T]$, and a Poisson random measure $N$ with $N(\ud \zeta, \ud t)\in Z \times [0,T] $, $Z \subseteq \R$, having nonnegative intensity $\nu(\ud \zeta)\ud t$. The measure $\nu(\ud \zeta)$ defined on a measurable space $(Z, \mathcal Z)$, is  $\sigma$-finite. The corresponding compensated random measure $\tilde N$, given by
\begin{equation} \label{def:cm}
\tilde N(\ud \zeta, \ud t)=N(\ud \zeta, \ud t)-\nu(\ud \zeta)\ud t,
\end{equation}
is a martingale measure with respect to the filtration $\bF$ and the historical probability $\P$.

We suppose that the savings account $S^0$ and the primary security account $S^1$ depend on some  stochastic factor $X=\{X_t,\ t \in [0,T]\}$ whose dynamics is given by a jump diffusion, and consider the following system of stochastic differential equations (in short SDEs)
\begin{equation}\label{eq:sistema}
\left\{
\begin{aligned}
\ud X_t&= b_0(t, X_t) \ud t + \sigma_0(t,X_t) \ud U_t + \int_Z K_0(\zeta;t, X_{t^-}) \tilde N(\ud \zeta; \ud t), \quad X_0=x\in \R, \\
\ud S^0_t&=S^0_t r(t, X_t)\ud t, \quad S^0_0=1,\\
\ud S^1_t&=S^1_{t^-}\left(b_1(t, X_t, S^1_t)\ud t+\sigma_1(t, X_t, S^1_t)\ud W_t+\int_Z K_1(\zeta;t, X_{t^-}, S^1_{t^-})\tilde N(\ud \zeta, \ud t)\right),\quad S^1_0>0,
\end{aligned}
\right.
\end{equation}
for every $t \in [0,T]$, where the coefficients $r(t,x), b_0(t,x), b_1(t,x,s^1), \sigma_0(t,x), \sigma_1(t,x,y)$, $K_0(\zeta;t,x)$ and $K_1(\zeta; t,x,s^1)$ are $\R$-valued measurable functions of their arguments such that a unique strong solution for the system \eqref{eq:sistema} exists, see for instance~\cite{OS}. In particular this implies that the triplet $(X, S^0, S^1)$ is an $\bF$-Markov process.
Observe that the savings account $S^0$ %is a {\em savings account} and
is assumed to be locally risk-free, which means that $S^0$ has finite variation. Moreover, here we assume %it is assumed
that the interest rate process $r=\{r(t, X_t),\ t \in[0,T]\}$ is $\bF$-adapted and nonnegative.
To specify some minimum requirements, we assume that the following integrability conditions hold $\P$-a.s.
\begin{align}
\int_0^T&\left(|b_0(t, X_t)|+\sigma_0^2(t, X_t)+\int_Z |K_0(\zeta;t, X_t)|\nu(\ud \zeta)\right)\ud t <  \infty,\label{integrab}\\
\int_0^T&\left(|r(t, X_t)|+|b_1(t, X_t, S^1_t)|+\sigma_1^2(t, X_t, S^1_t)+\int_Z |K_1(\zeta;t, X_t, S^1_t)|\nu(\ud \zeta)\right)\ud t <  \infty.\label{integrab2}
\end{align}
In addition, to ensure non-negativity, we suppose that $K_1(\zeta;t, X_{t^-}, S^1_{t^-}) + 1 > 0$ $\P$-a.s. for every $t\in [0,T]$.
Hence, by the Dol\'{e}ans-Dade exponential formula, the dynamics of $S^1$ can be written as
\begin{equation*}
S^1_t=S^1_0 e^ {Y^1_t} \quad \mbox{for every } \ t \in [0,T],
\end{equation*}
where the log-return process $Y^1=\{Y_t^1,\ t \in [0,T]\}$ satisfies
\begin{equation*}
\begin{aligned}
\ud Y^1_t  & = \left(b_1(t,X_t, S^1_t) - \frac{1}{2} \sigma_1^2(t, X_t, S^1_t)- \int_Z K_1(\zeta;t, X_{t}, S^1_t)\nu(\ud \zeta) \right ) \ud t+  \sigma_1(t, X_t, S^1_t) \ud W_t \\
& \qquad +\int_Z \log( 1+  K_1(\zeta;t, X_{t^-}, S^1_{t^-}) ) N(\ud \zeta, \ud t),
\end{aligned}
\end{equation*}
for every $t \in [0,T]$.
Note that condition  \eqref{integrab2} implies that $Y^1$ is an $(\bF,\P)$-semimartingale.

Recall that $\mathcal{V}^+_1$ denotes the set of strictly positive, finite, self-financing portfolios with initial capital equal to 1 in the given market model. %associated to $\bF$-strategies in the market.
Given $S^\delta \in \mathcal{V}^+_1$, define the corresponding vector of fractions or portfolio weights, $\pi_\delta=\{\pi_\delta(t)=(\pi_\delta^0(t),\pi_\delta^1(t))^\top,\ t \in [0,T]\}$ with
$$
\pi_\delta^1(t):=\frac{\delta^1_t S^1_{t^-}}{S^\delta_{t^-}} \quad \forall t \in [0,T]
$$
and where $\pi_\delta^0(t):=1-\pi_\delta^1(t)$ is the residual function invested in the savings account at time $t$.

Consequently, the behaviour of a portfolio $S^\delta \in \mathcal{V}^+_1$ is given by the SDE
\begin{align}
\ud S^\delta_t=&S^\delta_{t^-}\Big\{\left[r(t, X_t)+\pi_\delta^1(t)(b_1(t, X_t, S^1_t)-r(t, X_t))\right]\ud t + \pi_\delta^1(t)\sigma_1(t,X_t, S^1_t)\ud W_t \nonumber \\
 & \qquad +\pi_\delta^1(t)\int_Z K_1(\zeta;t, X_{t^-}, S^1_{t^-})\tilde N(\ud \zeta, \ud t)\Big\},\label{eq:posportfolio}
\end{align}
for every $t\in [0,T]$. Observe that, as a consequence of the Dol\'{e}ans-Dade exponential formula, strict positivity of $S^\delta$ is equivalent to the following condition
\begin{equation*} \label{Cla3}
1+\pi_\delta^1(t) \!\! \int_Z K_1(\zeta;t,X_{t^-}, S^1_{t^-})  N(\ud \zeta, \{t\} ) =\!\!\!\int_Z [ 1+\pi_\delta^1(t) K_1(\zeta;t, X_{t^-}, S^1_{t^-}) ] N(\ud \zeta, \{t\} ) > 0 \quad \P-{\rm a.s.} .
\end{equation*}

\noindent \subsection{Num\'{e}raire Portfolio}

As pointed out in Section \ref{sec:numeraire}, the benchmark approach employs a very special portfolio $S^{\delta_*} \in \mathcal{V}^+_1$ as benchmark. This is in several ways the {\em best} performing strictly positive portfolio. It is important to stress the growth optimality of the num\'{e}raire portfolio. Indeed, the num\'{e}raire portfolio and the so-called growth optimal portfolio (in short GOP), which is a portfolio having maximal growth rate, coincide whenever they exist, see e.g. Proposition 2.1 of~\cite{hs}.
This allows the explicit determination of the num\'{e}raire portfolio in the given market model by looking for its GOP, as we will see in the following.
\begin{definition} \label{def:logdrift}
The growth rate $g^\delta$ of a portfolio $S^\delta\in \mathcal{V}^+_1$ is defined as the infinitesimal drift of the SDE satisfied by $\log(S^\delta)$.
\end{definition}
%\noindent Goal: Finding the growth optimal portfolio.\\
%\noindent A GOP is a (strictly) positive self-financing portfolio $S^{\delta_*}$ which maximizes the portfolio growth rate, that is, the drift of its logarithm.
\noindent Given $S^\delta \in \mathcal{V}^+_1$ then by It\^o's formula it follows that
\begin{equation}\label{eq:log}
\begin{aligned}
\ud(\log (S^\delta_t)) &= g^\delta(t)\ud t + \pi^1_\delta(t)\sigma_1(t, X_t, S^1_t)\ud W_t \\
&\quad + \int_Z  \log\left(1+\pi^1_\delta(t)K_1(\zeta;t, X_{t^-}, S^1_{t^-})\right) \I_{\{ 1 + \pi^1_\delta(t)K_1(\zeta;t, X_{t^-}, S^1_{t^-})> 0\}} (\zeta)\tilde N(\ud \zeta, \ud t)
\end{aligned}
\end{equation}
where
\begin{equation} \label{eq:logdrift}
\begin{aligned}
&g^\delta(t):=r(t, X_t)+(b_1(t, X_t, S^1_t)-r(t, X_t))\pi^1_\delta(t)-\frac{1}{2}(\pi^1_\delta(t)\sigma_1(t, X_t, S^1_t))^2\\
&+\int_Z \left\{-\pi^1_\delta(t)K_1(\zeta;t, X_{t^-}, S^1_{t^-})+ \log\left(1+\pi^1_\delta(t)K_1(\zeta;t,X_{t^-}, S^1_{t^-})\right)  \I_{\{ 1 + \pi^1_\delta(t)K_1(\zeta;t, X_{t^-}, S^1_{t^-})> 0\}} (\zeta)\right\} \nu(\ud \zeta),
\end{aligned}
\end{equation}
for every $t\in [0,T]$.

\noindent We can now obtain an explicit characterization of the GOP in the given market model by means of the market price of risk. Denote by $L^1(Z)$ the set of all processes $K(\zeta;t)$ such that $\ds \int_Z |K(\zeta;t)|\nu(\ud \zeta)<\infty$ $\P$-a.s., for every $t\in[0,T]$.
Similarly to~\cite{cp2005}, we define the set $\mathcal B$ that, for every $t\in [0,T]$ is given by
\begin{gather*}
\mathcal B:=\Big\{(\sigma(t),K(\zeta;t)) \in \R \times L^1(Z)| \ (\sigma(t),K(\zeta;t))=\pi^1_\delta(t)(\sigma_1(t,X_t, S^1_t),K_1(\zeta;t, X_{t^-}, S^1_{t^-}))\Big\},
\end{gather*}

where $\pi^1_\delta=\{\pi_\delta^1(t), \ t \in [0,T] \}$ is an $\bF$-predictable process defining a strategy.
%\noindent \underline{Question}: at this stage, is it correct to consider $\bF$-strategies instead of $\bF^S$-strategies? Note that in general $\delta^1(t)=\pi^1_\delta(t)\cdot\frac{S^\delta(t)}{S^1(t)}$ is not $\bF^S$-predictable. \\
Since Assumption \ref{ass:numport} is in force, then strong forms of arbitrage are excluded in this setting and by Theorem 2.9 in~\cite{cp2005}, we know that there exists a continuous linear functional $\Gamma : \mathcal B \to \R$ such that
\begin{equation} \label{def:gamma}
\Gamma(\sigma_1(t,X_t, S^1_t),K_1(\zeta;t, X_{t^-}, S^1_{t^-}))=b_1(t, X_t, S^1_t)-r(t, X_t).
\end{equation}
Such a functional will be called {\em risk premium functional} and being continuous on a subset of $\R \times L^1(Z)$, $\Gamma$ can be represented by the processes $\theta=\{\theta(t,X_t, S^1_t)=\theta_1(t, X_t, S^1_t),\ t \in [0,T]\}$ and $\psi_\theta=\{\psi_\theta(\zeta;t, X_{t^-}, S^1_{t^-}),\ t \in [0,T]\}$ such that
\begin{equation} \label{gammarepre}
\begin{aligned}
&\Gamma \left(\sigma_1(t, X_t, S^1_t), K_1(\zeta;t, X_{t^-}, S^1_{t^-})\right)\\
&\qquad \qquad \qquad \qquad =\theta_1(t,X_t, S^1_t ) \sigma_1(t,X_t, S^1_t)+\int_Z K_1(\zeta;t,X_{t^-}, S^1_{t^-}) \psi_\theta(\zeta;t, X_{t^-}, S^1_{t^-})\nu(\ud \zeta),
\end{aligned}
\end{equation}
where $\theta_1$ is an $\bF$-predictable process that is assumed to be square-integrable in $t$. The process $\psi_\theta$ is  $\bF$-predictable and such that
$\int_Z | K_1(\zeta;t, X_{t^-}, S^1_{t^-}) \psi_\theta(\zeta;t, X_{t^-}, S^1_{t^-})| \nu(\ud \zeta) < \infty$, $ \P$-a.s. for every $t\in [0,T]$ and
\begin{equation}\label{eq:integrabilitapsitheta}
\int_0^T  \int_Z \left|\psi_\theta(\zeta;t, X_{t^-}, S^1_{t^-})\right| \nu(\ud \zeta)\ud t < \infty \quad \P-a.s..
\end{equation}
%To be coherent with the previous notation we put
%\[
%\theta_1(t, X_t,S^1_t)= \theta_1(t) \quad \textrm{and}\quad \psi_\theta(\zeta;t,X_t,S^1_t)=\psi_\theta(t).
%\]

\noindent Any vector $(\theta_1, \psi_\theta)$ satisfying \eqref{gammarepre}, will be called a {\em market price of risk representation}. In particular %here
$\theta_1$ is interpreted as the market price of diffusion risk and $\psi_\theta$ as the market price of jump risk. An economy of risk-neutral agents corresponds to pick $\theta_1=0$ and $\psi_\theta=0$ in \eqref{gammarepre}.
\begin{remark}
From no arbitrage it follows directly that (see Theorem 2.9 in \cite{cp2005})
$$
\int_Z K(\zeta;t) \psi_\theta(\zeta;t, X_{t^-}, S^1_{t^-})\nu(\ud \zeta) < \int_Z K(\zeta;t)\nu(\ud \zeta)
$$
for Lebesgue almost every $t \in [0,T]$ and every nonnegative $K$ appearing as second coordinate in $\mathcal B$. This means that if $\inf_Z K(\zeta;t)>F(t)>0$, for some deterministic process $F$, then
\begin{equation} \label{nuova}
\psi_\theta(\zeta;t, X_{t^-}, S^1_{t^-}) < 1, \quad  \nu(\ud \zeta) \otimes \ud t \otimes \P-\mbox{a.e.}.
\end{equation}
\end{remark}
\noindent As observed in~\cite{cp2005}, the risk premium functional $\Gamma$, satisfying \eqref{def:gamma}, is unique but the representation is not unique because $\Gamma$ is only defined on a subset of $L^1(Z)$. Any extension of this functional to the entire space of integrable
functions on $Z$ taking values in $(-1,+\infty)$, will have identical properties and can be represented by a distinct function $\psi_\theta$.

\noindent By Definition \ref{def:logdrift} and equations \eqref{eq:log}, \eqref{def:gamma} and \eqref{gammarepre}, now we can rewrite the given growth rate, see \eqref{eq:logdrift}, in terms of the risk premium functional:
\begin{equation*}
\begin{aligned}
g^\delta(t)&=r(t, X_t)+\pi^1_\delta(t)\theta_1(t, X_t, S^1_t)\sigma_1(t, X_t, S^1_t)-\frac{1}{2}\left(\pi^1_\delta(t)\sigma_1(t, X_t, S^1_t)\right)^2\\
&\quad \quad +  \pi^1_\delta(t)\int_Z K_1(\zeta;t, X_{t^-}, S^1_{t^-})(\psi_\theta(\zeta;t, X_{t^-}, S^1_{t^-})-1)\nu(\ud \zeta)\\
& \quad \quad \quad +\int_Z \log\left(1+\pi^1_\delta(t)K_1(\zeta;t, X_{t^-}, S^1_{t^-})\right)  \I_{\{ 1 + \pi^1_\delta(t)K_1(\zeta;t, X_{t^-}, S^1_{t^-})> 0\}} (\zeta) \nu(\ud \zeta)\quad \forall t \in [0,T].
\end{aligned}
\end{equation*}
To find the maximum value of $g^\delta$ on $\R$ one starts out by finding a stationary point and by differentiating with respect to $\pi^1_\delta$. We restrict ourself to consider  strategies $\pi^1_\delta$  such that
\begin{equation} \label{nuova1}
1 +  \pi^1_\delta(t)  K_1(t; \zeta, X_{t^-}, S^1_{t^-}) > 0, \quad \nu(\ud \zeta) \otimes \ud t \otimes \P-\mbox{a.e.}.
\end{equation}
This leads to the following first-order condition:
%This leads to necessary conditions. Then, by checking the supermartingale property, we have that such necessary conditions are also sufficient.
\begin{equation} \label{eq:firstorder}
\begin{split}
&\frac{\ud (g^\delta(t))}{\ud \pi^1_\delta(t)}=\theta_1(t, X_t, S^1_t)\sigma_1(t, X_t, S^1_t)+\int_Z K_1(\zeta;t, X_{t^-}, S^1_{t^-})(\psi_\theta(\zeta;t, X_{t^-}, S^1_{t^-})-1)\nu(\ud \zeta)\\
&\qquad \qquad -\pi^1_\delta(t)\sigma_1^2(t, X_t, S^1_t) +\int_Z K_1(\zeta;t, X_{t^-}, S^1_{t^-})\frac{1}{1+\pi^1_\delta(t)K_1(\zeta;t, X_{t^-}, S^1_{t^-})}  \nu(\ud \zeta) =0 \qquad \forall t \in [0,T].
\end{split}
\end{equation}
By Theorem 2.6 in~\cite{cp2005}, we know that \eqref{eq:firstorder} is a sufficient condition for a  portfolio $S^\delta \in \mathcal{V}^+_1$ to be the GOP in the underlying market. It is easy to check that a reasonable candidate is the couple
%By arguing as in the proof of Theorem 2.6 in~\cite{cp2005}, it is possible to prove that a reasonable candidate is the couple
\begin{equation} \label{max}
\left(\pi^1_{\delta_*}(t)\sigma_1(t,X_t, S^1_t), \pi^1_{\delta_*}(t)K_1(\zeta;t, X_{t^-}, S^1_{t^-})\right)=\left(\theta_1(t,X_t, S^1_t), \frac{\psi_\theta(\zeta;t, X_{t^-}, S^1_{t^-})}{1-\psi_\theta(\zeta;t, X_{t^-}, S^1_{t^-})}\right) \in \mathcal B,
\end{equation}
for almost every $t\in[0,T]$.
%\begin{remark}
%We should check the proof in detail. In particular they claim they use the density assumption on the set $\mathcal B$ but it is not clear where they introduce this kind of hypothesis.
%\end{remark}
\noindent Thus, by \eqref{def:gamma}, \eqref{gammarepre} and plugging \eqref{max} into \eqref{eq:posportfolio}, it is possible to derive the following dynamics:
\begin{equation}\label{GOPdyn}
\begin{split}
\ud S^{\delta_*}_t = & S^{\delta_*}_{t^-}\bigg\{\left[r(t, X_t)+\theta_1^2(t, X_t, S^1_t)+\int_Z \frac{\psi_\theta^2(\zeta;t, X_{t^-}, S^1_{t^-})}{1-\psi_\theta(\zeta;t, X_{t^-}, S^1_{t^-})}\nu(\ud \zeta)\right]\ud t \\
& \qquad \qquad \qquad \qquad \qquad + \theta_1(t, X_t, S^1_t)\ud W_t  + \int_Z \frac{\psi_\theta(\zeta;t, X_{t^-}, S^1_{t^-})}{1-\psi_\theta(\zeta;t, X_{t^-}, S^1_{t^-})}\tilde N(\ud \zeta;\ud t)\bigg\}
\end{split}
\end{equation}
for every $t \in [0,T]$. Let us observe that $ S^{\delta_*}$ is strictly positive if and only if
$$ \int_Z \frac{1}{1- \psi_\theta(\zeta; t, X_{t^-}, S^1_{t^-})} N(\ud \zeta, \{t\}) > 0, \quad \P-\mbox{a.s.}$$
for each $t \in [0,T]$ and this condition is implied by \eqref{nuova}.

\medskip

\noindent \subsubsection{Supermartingale property}
We now check that $S^0$ and $S^1$ are $(\bF,\P)$-supermartingales when they are denominated in units of the num\'{e}raire portfolio $S^{\delta_*}$. We start with the savings account $S^0$. By The product rule, for every $t \in [0,T]$, we have
\begin{align}
\ud(\hat S^0_t)&=\ud \left(\frac{S^0_t}{S^{\delta_*}_t}\right)=\frac{1}{S^{\delta_*}_{t^-}} \ud S^0_t+ S^0_t \ud \left(\frac{1}{S^{\delta_*}_t}\right)%\nonumber\\
%&=
=-\hat S^0_{t^-}\left[\theta_1(t, X_t, S^1_t) \ud W_t + \int_Z \psi_\theta(\zeta;t, X_{t^-}, S^1_{t^-}) \tilde N(\ud \zeta;\ud t)\right]\label{eq:s0}.
\end{align}
In other terms $\hat S^0$ can be written as the Dol\'{e}ans-Dade exponential of the process $\hat Y^0$
\[
\hat Y^0_t := -\int_0^t \theta_1(u, X_u, S^1_u) \ud W_u - \int_0^t \int_Z \psi_\theta(\zeta;u, X_{u^-}, S^1_{u^-}) \tilde N(\ud \zeta;\ud u)
\]
which is a $(\bF, \P)$-local martingale since $\theta_1$ is assumed to be square integrable and $\psi_\theta$ satisfies   \eqref{eq:integrabilitapsitheta} and \eqref{nuova}.

Then $\hat S^0=\mathcal{E}(\hat Y^0)$ is a nonnegative $(\bF, \P)$-local martingale.

We now look at the primary security account dynamics. By the integration by parts formula, for every $t \in [0,T]$ we get:
\begin{align*}
\ud(\hat S^1_t)&=\ud \left(\frac{S^1_t}{S^{\delta_*}_t}\right)=\frac{1}{S^{\delta_*}_{t^-}} \ud S^1_t+ S^1_{t^-} \ud \left(\frac{1}{S^{\delta_*}_t}\right)+\ud \left[S^1,\frac{1}{S^{\delta_*}}\right]_t\\
&=\hat S^1_{t^-}\left\{b_1(t, X_t, S^1_t)\ud t + \sigma_1(t, X_t, S^1_t)\ud W_t + \int_Z K_1(\zeta;t, X_{t^-}, S^1_{t^-}) \tilde N(\ud \zeta;\ud t) \right\}\\
& \quad -\hat S^1_{t^-}\left\{r(t, X_t) + \theta_1(t,X_t, S^1_t)\ud W_t + \int_Z\psi_\theta(\zeta;t, X_{t^-}, S^1_{t^-})\tilde N(\ud \zeta;\ud t)\right\}\\
& \qquad - \hat S^1_{t^-} \sigma_1(t, X_t, S^1_t)\theta_1(t, X_t,S^1_t)\ud t - \hat S^1_{t^-}\int_Z K_1(\zeta;t, X_t, S^1_{t^-})\psi_\theta(\zeta;t, X_{t^-}, S^1_{t^-})\nu(\ud \zeta) \ud t\\
& \qquad \quad -  \hat S^1_{t^-}\int_Z K_1(\zeta;t, X_{t^-}, S^1_{t^-})\psi_\theta(\zeta;t, X_{t^-}, S^1_{t^-})\tilde N(\ud \zeta;\ud t).
\end{align*}
By \eqref{def:gamma} and \eqref{gammarepre} the finite variation terms vanish and we obtain
\begin{equation}\label{eq:s1}
\ud(\hat S^1_t)=\hat S^1_{t^-}\Big\{(\sigma_1(t, X_t, S^1_t)-\theta_1(t, X_t, S^1_t))\ud W_t +  \int_Z K_{\theta}(\zeta;t, X_{t^-}, S^1_{t^-})\tilde N(\ud \zeta;\ud t)\Big\},
\end{equation}
for every $ t \in [0,T]$, where
\[
K_{\theta}(\zeta;t, X_{t^-}, S^1_{t^-}):=K_1(\zeta;t, X_{t^-}, S^1_{t^-})-\psi_\theta(\zeta;t, X_{t^-},S^1_{t^-}) -K_1(\zeta;t, X_{t^-}, S^1_{t^-})\psi_\theta(\zeta;t, X_{t^-}, S^1_{t^-})
.\]
Nothe that
\begin{equation}\label{eq:nonnegktheta}
1+K_{\theta}(\zeta;t, X_{t^-}, S^1_{t^-})>0 \quad \P-a.s. \quad \mbox{ for every } \ t \in [0,T]
\end{equation}
since for every $t \in [0,T]$
\[
1+K_{\theta}(\zeta;t, X_{t^-}, S^1_{t^-})=\big(1+K_1(\zeta;t, X_{t^-}, S^1_{t^-})\big)\big(1-\psi_{\theta}(\zeta;t, X_{t^-}, S^1_{t^-}) \big).
\]
Then $\hat S^1$ is the Dol\'{e}ans-Dade exponential of the process $\hat Y^1$ given by
\[
\hat Y^1_t:=\int_0^t(\sigma_1(u, X_u, S^1_u)-\theta_1(u, X_u, S^1_u))\ud W_u +  \int_Z K_{\theta}(\zeta;u, X_{u^-}, S^1_{u^-})\tilde N(\ud \zeta;\ud u)
\]
which is an $(\bF, \P)$-local martingale since
\[
\int_0^T(\sigma_1(t, X_t, S^1_t)-\theta_1(t, X_t, S^1_t))^2 \ud t \leq 2 \int_0^T \left\{\sigma_1^2(t, X_t, S^1_t)+\theta_1^2(t, X_t, S^1_t)\right\}\ud t <\infty \quad \P-a.s.
\]
\begin{align*}
& \int_0^T \int_Z |K_{\theta}(\zeta;t, X_{t^-}, S^1_{t^-})|\nu(\ud \zeta)\ud t \leq \int_0^T  \int_Z |K_1(\zeta;t, X_{t^-}, S^1_{t^-})|\nu (\ud \zeta) \ud t  \\
&+ \int_0^T \!\! \int_Z \!\! |\psi_{\theta}(\zeta;t, X_{t^-}, S^1_{t^-})|\nu (\ud \zeta) \ud t +\int_0^T \!\! \int_Z \!\!|K_1(\zeta;t, X_{t^-}, S^1_{t^-})\psi_{\theta}(\zeta;t, X_{t^-}, S^1_{t^-})|\nu(\ud \zeta)\ud t<\infty \quad \P-a.s..
\end{align*}

This implies that $\hat S^1=\mathcal{E}(\hat Y^1)$ is a nonnegative $(\bF, \P)$-local martingale.

Then, the dynamics given in \eqref{GOPdyn} has the property that benchmarked prices become $(\bF,\P)$-supermartingales. Thanks to Lemma 2.5 in~\cite{cp2005}, we can conclude that the given candidate $S^{\delta_*}$ is growth optimal.

\subsection{The $\bH$-benchmarked risk minimizing strategy}

We recall that the information available to traders is given by the filtration $\bH$ such that $\H_t \subseteq \F_t$ for every $t \in [0,T]$, and that in virtue of Assumption \ref{ass:misurabilitaS_T}, $\hat S_T$ is $\H_T$-measurable.
This model covers a great variety of situations. For example we may consider the case where agents can observe the prices but not the stochastic factor $X$ which influences their dynamics, or also the case where agents have information about the prices only at discrete times.

The process $X$ in \eqref{eq:sistema} may represent, for example, the trend of a correlated market, some macroeconomics factor or microstructure rule that drives the market.

The goal of this section is to characterize the benchmarked portfolio value in terms of a suitable function $g$ that solves a certain problem, and then to perform explicitly the optimal strategy.

In this framework, we assume that the benchmarked contingent claim $\hat H_T$ has a Markovian structure, i.e.
\begin{equation}\label{eq:claim}
\hat H_T= \hat H (T, \hat S_T),
\end{equation}
where $\hat H (t, \hat s)$ is a deterministic function, $\hat s$ denote the two-dimensional vector  $\hat s= (\hat s_0, \hat s_1)$ and finally and $\hat S$ is the $\R^2$-valued process given by $\hat S=(\hat S^0, \hat S^1)$. We also introduce the notation $ s$ for the two-dimensional vector $s= (s_0,  s_1)$ and $S$ for the two-dimensional vector process $S=( S^0, S^1)$.

\medskip

To ensure the correct mathematical tractability of the problem,
we make the following  assumption.
\begin{ass}\label{ass:welldefn}
The processes $\theta_1$, $\psi_\theta$, $K_\theta$ are such that strong solutions of the SDEs \eqref{eq:s0} and \eqref{eq:s1} exist.
\end{ass}

Set
\begin{gather*}
D_0(\omega, t):=\{\zeta \in Z: K_0(\zeta; t, X_{t^-}(\omega))\neq 0\}, \\
D_1(\omega,t):=\{\zeta \in Z: K_1(\zeta; t, X_{t^-}(\omega), S^1_{t^-}(\omega))\neq 0\},
\end{gather*}
for all $(\omega,t)\in \Omega \times [0,T]$. Here we are interested in the case where $D_0(\omega,t)\cap D_1(\omega, t)\neq \emptyset$, for all $(\omega,t)\in \Omega \times [0,T]$, which means that $X$ and $S^1$ may have common jump times.

This particular feature of the model has a financial motivation. Indeed, in this way we can also take into account the possibility of catastrophic events which may affect, at the same time, the prices and the hidden component that influences the market.

Define the function
\begin{equation}\label{funz:lambda}
\lambda_1(\omega, t):=\nu(D_1(\omega, t)), %\quad {\rm and} \quad \phi(\omega, t,\ud z):=\int_{D^1(\omega, t)}\delta_{K_1(t,X_{t^-}(\omega),S^1_{t^-}(\omega);\zeta)}(\ud z) \nu(\ud \zeta),
\end{equation}
for all $(\omega,t)\in \Omega \times [0,T]$.
Then %the $\ds (\P,\mathcal{F}_t)$-local characteristics of the integer valued counting measure $m(\ud t, \ud z)$, given by
%\begin{equation}\label{caratt.locali}
%(\lambda_t \ud t, \phi_t(\ud z))=(\lambda(t, X_{t^-},Y_{t^-})\ud t, \phi(t, X_{t^-},Y_{t^-},\ud z)),
%\end{equation}
$ \lambda_1:=\{\lambda_1( t), \ t\in [0,T]\} $ provides the $\bF$-predictable intensity of the point process $N^1_t$ which counts the total number of jumps of $S^1$ until time $t$ (see ~\cite{cg} and ~\cite{C} for the proof).

Analogously, we can define the function
\begin{equation}\label{funz:lambda0}\lambda_0(\omega, t):= \nu(D_0(\omega,t)),\end{equation}
for all $(\omega,t)\in \Omega \times [0,T]$; then, $\lambda_0:=\{ \lambda_0(t), \ t \in[0,T]\} $ provides the $\bF$-predictable intensity of the point process $N^X_t $ which counts the total number of jumps of $X$ until time $t$. The following conditions imply that both of the processes $N^1$ and $N^X$ are non-explosive and integrable (see e.g.~\cite[Chapter 3, Theorem T8]{Br}):
\begin{equation}\label{hp:nonexplosion}
\esp{\int_0^T\left\{\lambda_1(t)+ \lambda_0(t)\right\} \ud t}<\infty.
\end{equation}
Furthermore, we introduce the stochastic set $D_\theta$ as follows
 \begin{equation}\label{eq:Dtheta}
 D_\theta(\omega, t):= \{\zeta \in Z: \psi_\theta(\zeta; t, X_{t^-}(\omega), S^1_{t^-}(\omega))\neq 0 \ {\rm and} \ K_\theta(\zeta; t, X_{t^-}(\omega), S^1_{t^-}(\omega)) \neq 0 \},
 \end{equation}
 for every $(\omega, t) \in \Omega, \times [0,T]$. Since we will compute conditional expectation we make the following assumption.
\begin{ass}\label{ass:strategia}
The following integrability conditions hold:
\begin{gather*}
\esp{\int_0^T \left\{|b_0(t, X_t)|+ \sigma_0^2(t, X_t) + \int_{Z}| K_0(\zeta; t, X_{t^-})|^2  \nu(\ud \zeta)\right\}\ud t }<\infty,
\end{gather*}
\begin{gather*}
\esp{\int_0^T \left\{|b_1(t, X_t, S^1_t)|+ \sigma_1^2(t, X_t, S^1_t) + \int_{Z}| K_1(\zeta; t, X_{t^-}, S^1_{t^-})|^2  \nu(\ud \zeta)\right\}\ud t }<\infty,
\end{gather*}
\begin{gather*}
\esp{\int_0^T \left\{ r(t, X_t) + \theta_1^2(t, X_t, S^1_t) + \int_{Z}| \psi_\theta( \zeta; t, X_{t^-}, S^1_{t^-})|^2  \nu(\ud \zeta)\right\} \ud t} <\infty, \quad
%\end{gather*}
%and finally
%\begin{gather}\label{eq:integrabDtheta}
\esp{\int_0^T \nu(D_\theta(t))\ud t}<\infty.
\end{gather*}
\end{ass}

\begin{comment}
For the sake of simplicity, we will put $\phi:=\phi(t, x, s, \hat s)$ for every function $\phi$ defined on $[0,T]\times \R \times \R_+^2 \times \R_+^2$ omitting the dependence on the variables.

\begin{gather*}
r=r(t,x), \quad b_0= b_0(t,x), \quad \sigma_0=\sigma_0(t,x)\quad K_0(\zeta)= K_0(\zeta; t, x)\\
b_1= b_1(t,x, s_1), \quad \sigma_1=\sigma_1(t,x, s_1)\quad K_1(\zeta)= K_1(\zeta; t, x, s_1)\\
\theta_1= \theta_1(t,x,s_1), \quad \psi_{\theta}(\zeta)= \psi_\theta(\zeta;t,x,s_1), \quad K_\theta(\zeta)=K_\theta(\zeta; t, x, s_1).
\end{gather*}
\end{comment}

We recall that $S=(S^0, S^1)$ and $\hat S=(\hat S^0, \hat S^1)$. Denote by $\C_b^{1,2,2,2}([0,T]\times \R \times \R_+^2\times \R_+^2)$ the space of all bounded functions that are $\C^1$ with respect to the time variable $t$, $\C^2$ with respect to $x$,$s$ and $\hat s$, with bounded derivatives.
Then, the following result gives the Markovian structure of the triplet $(X, S, \hat S)$.
\begin{lemma}\label{lemma:generatore}
Under Assumptions \ref{ass:welldefn}, \ref{ass:strategia} and condition \eqref{hp:nonexplosion}, the triplet $(X, S, \hat S)$ is an $(\bF,\P)$-Markov process with generator
\begin{equation}\label{eq:generatore}
\L^{X,S, \hat S} f(t,x,s,\hat s) = \L_1 f(t,x,s,\hat s) + \L_2 f(t,x,s,\hat s) + \L_J f(t,x,s,\hat s),
\end{equation}
where
 \begin{align*}
\L_1 f(t,x,s,\hat s) := \frac{\partial f}{\partial t}(t,x,s,\hat s)+b_0(t,x)\frac{\partial f}{\partial x}(t,x,s,\hat s)+ r(t,x) s_0 \frac{\partial f}{\partial s_0}(t,x,s,\hat s)+  b_1(t,x,s_1) \, s_1 \frac{\partial f}{\partial s_1}(t,x,s,\hat s),
\end{align*}
\begin{align*}
&\L_2 f(t,x,s,\hat s):= \frac{1}{2} \sigma_0^2(t,x) \frac{\partial^2 f}{\partial x^2}(t,x,s,\hat s) + \frac{1}{ 2} \sigma_1^2 (t,x,s_1)\, s_1^2  \frac{\partial^2 f}{\partial s_1^2}(t,x,s,\hat s) \\
& \qquad + \frac{1}{2} \theta_1^2(t,x,s_1) \, \hat s_0^2  \frac{\partial^2 f}{\partial \hat s_0^2}(t,x,s,\hat s) +\frac{1}{2} (\sigma_1(t,x,s_1)-\theta_1(t,x,s_1))^2\, \hat s_1^2  \frac{\partial^2 f}{\partial \hat s_1^2}(t,x,s,\hat s) \\
& \qquad+ \rho \, \sigma_0(t,x) \, \sigma_1(t,x,s_1) s_1  \frac{\partial^2 f}{\partial x\partial s}(t,x,s,\hat s)
- \rho \, \sigma_0(t,x,s_1) \, \theta_1(t,x,s_1) \hat s_0  \frac{\partial^2 f}{\partial x\partial \hat s_0}(t,x,s,\hat s)\\
 & \qquad + \rho \, \sigma_0(t,x,s_1) \, (\sigma_1(t,x,s_1)-\theta_1(t,x,s_1)) \hat s_1  \frac{\partial^2 f}{\partial x\partial \hat s_1}(t,x,s,\hat s)\\
&\qquad - \sigma_1(t,x,s_1) \theta_1(t,x,s_1) s_1 \hat s_0  \frac{\partial^2 f}{\partial s_1\partial \hat s_0}(t,x,s,\hat s)\\
&\qquad +  \sigma_1(t,x,s_1) (\sigma_1(t,x,s_1)-\theta_1(t,x,s_1)) s_1 \hat s_1  \frac{\partial^2 f}{\partial s_1\partial \hat s_1}(t,x,s,\hat s) \\
&\qquad - \theta_1(t,x,s_1) \, (\sigma_1(t,x,s_1)-\theta_1(t,x,s_1)) \hat s_1 \hat s_0  \frac{\partial^2 f}{\partial \hat s_0 \partial \hat s_1}(t,x,s,\hat s),
\end{align*}
\begin{align*}
& \L_J f(t,x,s,\hat s) := - \int_Z \left\{\frac{\partial f}{\partial x}(t,x,s,\hat s) K_0(\zeta;t,x) + \frac{\partial f}{\partial s_1}(t,x,s,\hat s) s_1 K_1(\zeta;t,x,s_1)\right\}\nu(\ud \zeta)\\
& \quad - \int_{Z}\left\{-\frac{\partial f}{\partial \hat s_0}(t,x,s,\hat s) \hat s_0 \psi_\theta(\zeta;t,x,s_1) - \frac{\partial f}{\partial \hat s_1}(t,x,s,\hat s) \hat s_1 K_\theta(\zeta;t,x,s_1)  \right\}\nu(\ud \zeta) +  \int_Z \Delta f (\zeta;t,x,s,\hat s)\nu(\ud \zeta),
\end{align*}
and
\[
\begin{split}
\Delta f (\zeta;t,x,s,\hat s)\! = & \! f \Big (t,x+K_0(\zeta;t,x),s_0, s_1( 1 +K_1(\zeta;t,x,s_1)),\hat s_0( 1 - \psi_{\theta}(\zeta;t,x,s_1)), \hat s_1 (1+K_\theta(\zeta;t,x,s_1)  ) \Big)\\
& - f(t,x,s_0, s_1, \hat s_0, \hat s_1 ).
\end{split}
\]
More precisely,  for any  function $f(t,x,s,\hat s) \in \C_b^{1,2,2,2}([0,T]\times \R \times \R_+^2\times \R_+^2)$, the following semimartingale decomposition holds
\begin{equation}  \label{eq:semimg_representation}
f(t,X_t,S_t, \hat S_t) = f(t,x_0,S_0, \hat S_0) + \int_0^t \L^{X,S,\hat S} f(r, X_r, S_r, \hat S_r) \ud r + M^f_t, \quad t \in [0,T],
\end{equation}
where $M^f=\{M_t^f,\ t \in [0,T]\}$ is the $(\bF,\P)$-martingale given by
\begin{align}
M^f_t:=& \int_0^t \frac{\partial f}{\partial x}(u, X_u, S_u, \hat S_u)\sigma_0(u, X_u) \ud U_u \nonumber \\
& + \int_0^t  \left\{\frac{\partial f}{\partial s_1}(u, X_u, S_u,\hat S_u) \sigma_1(u, X_u, S^1_r) S^1_u - \frac{\partial f}{\partial \hat s_0}(u, X_u, S_u, \hat S_u) \hat S^0_u \theta_1(u, X_u, S^1_u) + \right. \nonumber\\
& \left. \hspace{3cm} + \frac{\partial f}{\partial \hat s_1}(u, X_u, S_u, \hat S_u) \hat S^1_u (\sigma_1(u, X_u, S^1_u)-\theta_1(u, X_u, S^1_u)) \right\}  \ud W_u \nonumber \\
& + \int_0^t \int_Z  \Big (f (u,X_u, S_u, \hat S_u) -f(u,X_{u^-},S_{u^-}, \hat S_{u^-})\Big) \tilde N(\ud \zeta; \ud u).\label{eq:mf}
\end{align}

\end{lemma}
\begin{proof}
Thanks to Assumption \ref{ass:welldefn}, the process $(X,S, \hat S)$ is the unique solution of the martingale problem associated to the operator $\L^{X,S, \hat S}$ (see Theorem 3.3 in~\cite{KO}). Then, it is an $(\bF,\P)$-Markov process.

By applying It\^{o}'s formula to the function $f(t,x,s,\hat s) \in \C_b^{1,2,2,2}([0,T]\times \R \times \R^2\times \R^2)$, we get the expression \eqref{eq:semimg_representation}. Moreover, by Assumption \ref{ass:strategia} and condition \eqref{hp:nonexplosion}, the following integrability conditions hold

\begin{gather*}
\esp{\int_0^T \left|\L^{X,S,\hat S}f(t, X_t, S_t, \hat S_t)\right| \ud t} <\infty, \\
\esp{\int_0^T  \sigma^2_0(t, X_t) \left( \frac{\partial f}{\partial x}(t, X_t, S_t, \hat S_t) \right)^2\ud t} \leq B_f \ \esp{\int_0^T  \sigma^2_0(t, X_t)\ud t}< \infty,
\end{gather*}
\begin{gather*}
\begin{aligned}[t]
&\mathbb{E} \left[ \int_0^T \left\{\sigma_1(t, X_t, S^1_t) S^1_t  \frac{\partial f}{\partial s_1}(t, X_t, S_t, \hat S_t) - \theta_1(t, X_t, S^1_t) \hat S^0_t  \frac{\partial f}{\partial \hat s_0}(t, X_t, S_t, \hat S_t)  \right. \right.\\
&\quad \left. \left. + (\sigma_1(t, X_t, S^1_t)-\theta_1(t, X_t, S^1_t)) \hat S^1_t  \frac{\partial f}{\partial \hat s_1}(t, X_t, S_t, \hat S_t) \right\}^2 \ud t \right]\leq \\
&\quad \bar B_f \ \esp{ \int_0^T \!\! \left\{\sigma_1^2(t, X_t, S^1_t) (S^1_t)^2  + \theta_1^2(t, X_t, S^1_t) (\hat S^0_t)^2 + (\sigma_1^2(t, X_t, S^1_t)+\theta_1^2(t,X_t,S^1_t)) (\hat S^1_t)^2\right\} \ud t }< \infty,
\end{aligned}
\end{gather*}
for some suitable positive constants $B_f$, $\bar B_f$, and
\[
\begin{aligned}
\esp{\int_0^T\!\!\!\! \int_Z\! |f(t,X_t, S_t, \hat S_t )-  f(t,X_{t^-},S_{t^-}, \hat S_{t^-}) | \nu(\ud \zeta)\;\ud t}\! \leq 2 \|f\|_{\infty}\esp{\int_0^T\!\!\! \{ \lambda_0(t) + \lambda_1(t) + \nu(D_\theta(t))\} \ud t }\!<\infty,
\end{aligned}
\]
where $\lambda_1$ and $\lambda_0$ and $D_\theta$ are defined in \eqref{funz:lambda}, \eqref{funz:lambda0} and \eqref{eq:Dtheta} respectively. %and satisfy \eqref{hp:nonexplosion} and \eqref{eq:integrabDtheta}.
This means that all integrals in \eqref{eq:semimg_representation} are well-defined and $M^f$ is indeed an $(\bF,\P)$-martingale.
\end{proof}

Thanks to the Markov property there exists a measurable function $g(t,x,s,\hat s)$ such that
\[
g(t, X_t, S_t, \hat S_t)=\condespf{\hat H_T}.
\]

We recall that the notation ${}^pZ$ denotes the $\bH$-predictable projection of the process $Z$. The following proposition characterizes the benchmarked $\bH$-risk-minimizing strategy for the Markovian jump-diffusion market model considered in this section.

\begin{proposition}\label{prop:strategia_esplicita}
Suppose that Assumptions  \ref{ass:welldefn} and \ref{ass:strategia}  are in force and that
\eqref{hp:nonexplosion} and \eqref{eq:claim} hold. Set
\[g(t,X_t,S_t, \hat S_t):=\condespf{\hat H_T }, \quad \forall t \in [0,T].\]
If $g\in \mathcal{C}^{1,2,2,2}_b([0,T] \times \R \times \R_+^2\times \R_+^2 )$, then it solves the problem
\begin{gather}\label{eq:dir}
\left\{
\begin{aligned}
&\L^{X,S, \hat S}g(t,x,s, \hat s)=0\\
&g(T,x,s, \hat s)=\hat H (T, \hat s).
\end{aligned}
\right.
\end{gather}
Moreover, if $\ds \ {}^pa^{0,0}(t) {}^pa^{1,1}(t) - \left[{}^pa^{0,1}(t)\right]^2\neq 0$ for every $t \in [0,T]$, with $a^{0,0}(t), a^{1,1}(t)$ and $a^{0,1}(t)$ respectively given by \eqref{eq:Sigma0}, \eqref{eq:Sigma1} and \eqref{eq:Sigma2},
then the $\bH$-benchmarked risk-minimizing strategy $\phi^\H=(\eta^\H, \delta^\H)$  with $\delta^\H=(\delta^{\H,0},\delta^{\H,1})$ is explicitly given by

\begin{equation*}
\left\{
\begin{aligned}
&\delta^{\H,0}_t=
\frac{{}^pa^{0,0}(t) \ {}^ph_0(t)- {}^pa^{0,1}(t) \ {}^ph_1(t)}
{{}^pa^{0,0}(t) \ {}^pa^{1,1}(t) - \left[{}^pa^{0,1}(t)\right]^2 },\\
&\delta^{\H,1}_t=
\frac{ - {}^pa^{0,1}(t) \ {}^ph_0(t) + {}^pa^{1,1}(t) \ {}^ph_1(t)  }
{ {}^pa^{0,0}(t) \ {}^pa^{1,1}(t) - \left[{}^pa^{0,1}(t)\right]^2 },\\
&\eta_t^\H=\condesph{g(t,X_t, S_t, \hat S_t)}-\left(\delta_t^\H\right)^\top \cdot \condesph{\hat S_t},
\end{aligned}
\right.
\end{equation*}
for every $t \in [0,T]$,  where
\begin{align}
& h_0(t):= -\hat S^0_t \theta_1(t, X_t, S^1_t) \left\{\frac{\partial g}{\partial s_1}(t, X_t, S_t, \hat S_t) \ S^1_t \ \sigma_1(t, X_t, S^1_t) - \frac{\partial g}{\partial \hat s_0}(t, X_t, S_t, \hat S_t) \hat S^0_t \theta_1(t, X_t, S^1_t)\right.\nonumber\\
& \qquad  \left.+ \frac{\partial g}{\partial \hat s_1}(t, X_t, S_t, \hat S_t)(\sigma_1(t, X_t, S^1_t) -\theta_1(t, X_t, S^1_t))\hat S^1_t \right\}\nonumber\\
& \qquad - \hat S^0_t \rho \; \sigma_0(t, X_t) \theta_1 (t, X_t, S^1_t) \frac{\partial g}{\partial x}(t, X_t, S_t, \hat S_t)  - \hat S^0_{t} \int_Z\psi_\theta(\zeta; t, X_{t}, S^1_{t})\Delta g(\zeta; X_t, S_t, \hat S_t)\nu(\ud \zeta), \label{eq:h0}\\
 %\end{align}
 %\begin{align}
& h_1(t): = \hat S^1_t (\sigma_1(t,X_t,S^1_t)-\theta_1(t,X_t,S^1_t)) \left\{\frac{\partial g}{\partial s_1}(t, X_t,S_t,\hat S_t) S^1_t \sigma_1(t, X_t,S^1_t)\right. \nonumber\\
&\qquad \left. - \frac{\partial g}{\partial \hat s_0}(t, X_t,S_t,\hat S_t)\hat S^0_t \theta_1(t,X_t,S^1_t)+ \frac{\partial g}{\partial \hat s_1}(t, X_t,S_t,\hat S_t)(\sigma_1(t,X_t,S^1_t)-\theta_1(t,X_t,S^1_t))\hat S^1_t\right\} \nonumber\\
 &\qquad + \hat S^1_t \rho \; \sigma_0(t, X_t, S^1_t) (\sigma_1(t, X_t, S^1_t) -\theta_1(t, X_t, S^1_t))\frac{\partial g}{\partial x}(t, X_t,S_t,\hat S_t)\nonumber\\
 &\qquad + \hat S^1_{t} \int_Z K_\theta(\zeta; X_{t}, S^1_{t}) \Delta g(\zeta; X_t, S_t,\hat S_t)\nu(\ud \zeta)   \label{eq:h1}
\end{align}
with
\[
\begin{aligned}
\Delta g(\zeta; t, x,s,\hat s):=& g\Big(t,x+K_0(\zeta;t,x),s_0, s_1(1+ K_1(\zeta;t,x,s_1)), \hat s_0(1-\psi_\theta(\zeta;t,x,s_1)), \hat s_1(1+K_\theta(\zeta;t,x,s_1))\Big)\\
&- g(t,x,s_0,s_1,\hat s_0,\hat s_1),
\end{aligned}
\]
as in Lemma \ref{lemma:generatore}, and
\begin{gather}
   a^{0,0}(t):= \left(\hat S^0_{t}\right)^2 \left(\theta_1^2(t, X_t, S^1_t)  +  \int_{Z}   \psi_\theta^2 (\zeta; t, X_{t}, S^1_{t}) \nu (\ud \zeta)\right),\label{eq:Sigma0}\\
   a^{1,1}(t):= \left(\hat S^1_{t}\right)^2\left[(\sigma_1(t, X_t, S^1_t)-\theta_1(t, X_t, S^1_t))^2  +  \int_{Z} K_\theta^2(\zeta; t, X_{t}, S^1_{t}) \nu (\ud \zeta)\right]\label{eq:Sigma1}\\
   a^{0,1}(t):= - \hat S^0_{t} \hat S^1_{t}\! \left[\theta_1(t, X_t, S^1_t) (\sigma_1(t, X_t, S^1_t)-\theta_1(t, X_t, S^1_t))\! +\!\! \int_Z \!\! \psi_\theta(\zeta;t,X_{t}, S^1_{t})K_\theta(\zeta;t, X_{t}, S^1_{t}) \nu (\ud \zeta)\right]. \label{eq:Sigma2}
\end{gather}

\end{proposition}

\begin{proof}
First, let $g\in C^{1,2,2,2}_b\left([0,T]\times \mathbb{R} \times \mathbb{R}_+^2 \times \R_+^2\right)$ and apply It\^{o}'s formula. We get that for every $t \in [0,T]$
\begin{equation}\label{g}
g(t,X_t,S_t,\hat S_t)= g(0,X_0,S_0, \hat S_0)+ \int_0^t \L^{X,S,\hat S}  g(r,X_r,S_r,\hat S_r)\ud r + M^g_t,
\end{equation}
where $M^g=\{M_t^g,\ t \in [0,T]\}$ is an $(\bF, \P)$-martingale thanks to Assumption \ref{ass:strategia}. By definition, $g(t, X_t, S_t, \hat S_t)=\condespf{\hat H_T}$ for each $t \in [0,T]$ and therefore $g(t, X_t, S_t, \hat S_t)$ is an $(\bF,\P)$-martingale; this means that all finite variation terms in \eqref{g} vanish and implies that $g$ solves the problem \eqref{eq:dir}.

We now focus on the second part. Since $\hat S$ is an $(\bF,\P)$-local martingale, $\hat H_T$ can be decomposed as (see \eqref{GKWdecomp})
\begin{equation*}
\hat H_T= \tilde H_0 +\int_0^T (\delta^{\F}_r)^\top \cdot \ud \hat S_r + \tilde L_T^{\hat H}\quad \P-\mbox{a.s.}.
\end{equation*}

By Theorem \ref{prop:fs} in Section \ref{sec:brm}, we know that there exists a unique $\bH$-benchmarked risk-minimizing strategy $\delta^{\H}$ that corresponds to the integrand appearing in decomposition \eqref{eq:GKWpartial}. Moreover, it can be characterized in terms of $\delta^{\F}$ thanks to Proposition \ref{explicit_delta}.

%By Theorem \ref{prop:fs} the unique $\bH$bencmarked risk-minimizing strategy  is $\delta = \delta^H$, where $\delta^H$ is the integrand in the decomposition \eqref{eq:GKWpartial}, which is connected to decomposition \eqref{GKWdecomp} by \eqref{eq:deltah}.
We denote by $P$ the $(\bF,\P)$-martingale given by
\begin{equation}\label{V(delta*)}
P_t:=\condespf{\hat H_T}=g(t,X_t,S_t, \hat S_t),
\end{equation}
for every $t \in [0,T]$. Then, %by %the Galtchouk-Kunita-Watanabe decomposition of $\hat H_T$ under complete information,
%\eqref{GKWdecomp},
by the Galtchouk-Kunita-Watanabe decomposition of $\hat H_T$ under complete information, \eqref{V(delta*)} becomes
\begin{equation}
P_t=\tilde H_0+\condespf{\int_0^T  (\delta^{\F}_r)^\top \cdot \ud \hat S_r} + \condespf{ \tilde L_T^{\hat H}} = \tilde H_0+\int_0^t (\delta^\F_r)^\top \cdot \ud \hat S_r + \tilde L_t^{\hat H},  \label{relazioni-mg}
\end{equation}
for every $t \in [0,T]$. The process $\delta^{\F}$ can be computed in terms of the %$\bF$-predictable
density of the sharp bracket of $P$ and $\hat S$ with respect to the sharp bracket of  $\hat S$.
Indeed
\begin{align*}
\langle P,\hat S\rangle_t  =\left\langle\int (\delta^\F_r)^\top \cdot \ud \hat S_r,\hat S\right\rangle_t+\langle \tilde L^{\hat H},\hat S\rangle_t  =\left\langle \int (\delta^\F_r)^\top \cdot \ud \hat S_r,\int \ud \hat S_r \right\rangle_t=\int_0^t  (\delta^\F_r)^\top \cdot \ud \langle \hat S\rangle_r.
\end{align*}
Then for each component we get
\begin{align*}
\langle P,\hat S^0 \rangle_t  =\int_0^t  \delta^{\F,0}_r \cdot \ud \langle \hat S^0\rangle_r +\int_0^t  \delta^{\F,1}_r \cdot \ud \langle \hat S^0, \hat S^1\rangle_r
\end{align*}
and
\begin{align*}
\langle P,\hat S^1 \rangle_t  =\int_0^t  \delta^{\F,0}_r \cdot \ud \langle \hat S^0, \hat S^1\rangle_r + \int_0^t  \delta^{\F,1}_r \cdot \ud \langle \hat S^1\rangle_r.
\end{align*}
%which gives the expression for $\delta^\F=(\delta^{\F,0},\delta^{\F,1})$  as the solution of the linear system in terms of the Radon-Nikodym derivatives above.

Now observe that by the expression of the $(\bF,\P)$-martingale $M^g$ which is given by \eqref{eq:mf} replacing $f$ by $g$, and those of $\hat S^0$ and $\hat S^1$ we get that
\begin{eqnarray}
 \langle P,\hat S^0 \rangle_t= \int_0^t h_0(r) \ud r, \quad t \in [0,T],
\end{eqnarray}
where $h_0(t)$ is given by \eqref{eq:h0},
\begin{equation*}
  \langle \hat S^0\rangle_t= \int_0^t  (\hat S^0_{r})^2 \left\{  \theta_1^2(r, X_t, S^1_r)  +  \int_{Z}   \psi_\theta^2 (\zeta;r, X_{t}, S^1_{t}) {\nu}(\ud \zeta) \right\} \ud r =\int_0^t  a^{0,0}(r) \ud r,
\end{equation*}
for each $t \in [0,T]$, with $a^{0,0}(t)$ given by \eqref{eq:Sigma0},
\begin{eqnarray}
 \langle P,\hat S^1 \rangle_t= \int_0^t  h_1(r) \ud r, \quad t \in [0,T]
\end{eqnarray}
where $h_1(t)$ is given by \eqref{eq:h1}, while
\begin{equation*}
    \langle \hat S^1\rangle_t= \int_0^t  (\hat S^1_{r})^2 \left\{  (\sigma_1(r, X_r, S^1_r)-\theta_1(r, X_r, S^1_r))^2  +  \int_{Z} K^2_\theta(\zeta; r, X_{r}, S^1_{r})\nu(\ud \zeta) \right\} \ud r =\int_0^t a^{1,1}(r) \ud r,
\end{equation*}
for every $t \in [0,T]$, with $a^{1,1}(t)$ given by \eqref{eq:Sigma1} and finally
\begin{equation*}
\langle\hat S^0, \hat S^1\rangle_t= \int_0^t a^{0,1}(r) \ud r,
\end{equation*}
for every $t \in [0,T]$, with $a^{0,1}(t)$ given by \eqref{eq:Sigma2}.

Hence $\delta^\F=(\delta^{\F,0},\delta^{\F,1})^\top$ satisfies  $a(t) \ \delta^{\F}_t =  h(t)$ where $a$ is the matrix
\[
a(t):=\left(\begin{array}{cc}
a^{0,0}(t)& a^{0,1}(t)\\
a^{0,1}(t)& a^{1,1}(t) \end{array} \right)
\]
and $h$ is the vector $h(t)=(h_0(t), h_1(t))^\top$.
Finally we can characterize explicitly the benchmarked $\bH$-risk-minimizing strategy $\delta^\H$, by \eqref{eq:deltah}. More precisely, for every $t \in [0,T]$, we have

\[
\delta^{\H,0}_t=
\frac{{}^pa^{0,0}(t) \  {}^ph_0(t)- {}^pa^{0,1}(t) \ {}^ph_1(t)}
{{}^pa^{0,0}(t) \ {}^pa^{1,1}(t) - \left[{}^pa^{0,1}(t)\right]^2 }
\]
\[
\delta^{\H,1}_t=
\frac{ - {}^pa^{0,1}(t)  \ {}^ph_0(t) + {}^pa^{1,1}(t) \ {}^ph_1(t)  }
{ {}^pa^{0,0}(t) \ {}^pa^{1,1}(t) - \left[{}^pa^{0,1}(t)\right]^2 }
\]
Finally note that
\begin{gather*}
\eta_t^\H=\condesph{\hat H_T-\left( \delta^\H \right)^\top \cdot \hat S_t}=\condesph{\hat H_T}-\condesph{\left( \delta^\H \right)^\top \cdot \hat S_t}\\
=\condesph{\condespf{\hat H_T}}- \left( \delta^\H \right)^\top \cdot \condesph{\hat S_t} = \condesph{g(t,X_t, S_t, \hat S_t)}-\left(\delta_t^\H\right)^\top \cdot \condesph{\hat S_t},
\end{gather*}
and this concludes the proof.

\end{proof}

\begin{remark}
We consider now the case where investors can observe the dynamics of the saving account $S^0$, the primary security account $S^1$ and the num\'{e}raire portfolio $S^{\delta_*}$ but not the stochastic factor $X$ that affects their dynamics. This means that the filtration $\bH$ coincides with the filtration generated by the pair $(S, S^{\delta_*})$, i.e. $\H_t = \F^{S, S^{\delta*}}_t$ for every $t \in[0,T]$. Then, in this framework the computation of the benchmarked $\bH$-risk-minimizing strategy leads to a filtering problem with jump diffusion observations, where the signal process is given by the unobservable stochastic factor $X$ and the observation is the pair $(S, S^{\delta_*})$. The solution of the associated filtering problem allows us to provide an explicit representation of the optimal strategy in terms of the so-called filter. More precisely, at any time $t\in [0,T]$, the filter, defined as  $\pi_t(f) =  \condesph{f(t,X_t)}$ for any integrable process $f(t,X_t)$, provides the conditional law of the stochastic factor $X_t$ given the observed history $\H_t$. In particular, the $\bH$-predictable projection of any integrable process of the form $F(t,X_t,S_t,\hat S_t )$ can be computed in terms of the filter in the following way

$${}^pF(t,X_t,S_t,\hat S_t) = \pi_{t^-} (F(t,X_t, S_{t^-},\hat S_{t^-}))$$

where $\pi_{t^-}$ denotes the left version of $\pi_t$. We refer to~\cite{cco1}  and~\cite{cco2}  for the solution of the filtering problem in the framework of partially observed jump-diffusion systems.
\end{remark}

\subsection{GOP and Risk-neutral measures for a jump diffusion driven market model} \label{sec:rel}
The benchmarked risk minimization only requires the existence of the GOP  without making the restrictive assumption of the existence of a risk neutral probability measure, as in the classical risk minimization. In Section \ref{sec:jumpdiff} we investigated a Markovian jump-diffusion market model and we observed  that no-arbitrage implies the existence of a GOP for  such a model.  The purpose of this section is to discuss the relationship between the GOP and the existence of a martingale measure for a general jump-diffusion driven market model. More precisely, we will show that existence of  a martingale measure in fact implies the existence of  the GOP and that  the converse implication  is obtained  under additional integrability conditions (see Proposition \ref{prop:GOP} below).
We assume that the dynamics of the processes $S^0$ and $S^1$ are similar to that introduced in Section \ref{sec:jumpdiff}, see the second and third equations in system \eqref{eq:sistema}, but with general coefficients. In particular,

\begin{equation} \label{eq:riskless_gen}
\ud S^0_t=S^0_t r(t)\ud t,
\end{equation}
and
\begin{equation}\label{eq:risky_gen}
\ud S^1_t=S^1_{t^-}\left(b_1(t)\ud t+\sigma_1(t)\ud W_t+\int_Z K_1(\zeta;t)\tilde N(\ud \zeta, \ud t)\right)
\end{equation}
for $t \in [0,T]$, with $S^0_0=1$  and $S^1_0 > 0$.
Here we assume that the processes $r=\{ r(t),\ t \in[0,T]\}$ and $\sigma_1 = \{\sigma_1(t), \ t \in [0,T]\}$ are $\bF$-adapted and nonnegative, $b_1=\{b_1(t), \ t\in[0,T]\}$ is $\bF$-adapted and $K_1=\{K_1(\zeta; t), \ t \in [0,T]\}$ is $\bF$-predictable and that all the processes satisfy the integrability conditions stated in Section \ref{sec:jumpdiff}.

\begin{definition}
We say that a probability $\Q$ on $(\Omega, \F)$ is a martingale measure if  $\Q$ is locally equivalent to $\P$ $($i.e. $\Q|_{\F_t}$ is equivalent to
$\P |_{\F_t} )$
and the process $S^1e ^{-\int r(s) \ud s}$ is an $(\bF,\Q)$-local martingale.
\end{definition}
\noindent To this aim we need a suitable version of the Girsanov theorem, which we recall for reader's convenience.
\begin{theorem}\label{Gi}
Let $\xi=\{\xi(t), \ t \in [0,T]\}$ and $\eta(\zeta;\cdot)=\{\eta(\zeta;t),\ t \in [0,T]\}$ be $\bF$-predictable processes such that for any finite $t \in \R_+$, the following conditions hold:
\begin{equation}\label{altre}
 \int_0^t \xi(s)^2 \ud s < \infty, \quad 1+  \eta(\zeta ; t) > 0, \quad \int_0^t  \int_Z |\eta(\zeta ; s) + 1| \nu(\ud\zeta)\ud s < \infty, \quad \P-{\rm a.s.}
 \end{equation}
for every $t \in [0,T]$.
Define the process $L=\{L_t,\ t \in [0,T]\}$ as
 \begin{equation} \label{density}
\ud L_t=L_{t^-}\left(\xi(t)\ud W_t+\int_Z \eta(\zeta ; t)\tilde N(\ud \zeta, \ud t)\right)
\end{equation}
and suppose that for all finite $t \in \R_+$,  $\esp{L_t}=1.$
\noindent
Then, there exists a probability measure $\Q$ on $\bF$ locally equivalent to $\P$ with
$$
\left.\frac{\ud \Q }{ \ud \P}\right|_{\F_t} = L_t,\quad t \in [0,T]
$$
such that the process
$$
\ud W^{\mathbf Q}_t:=\ud W_t - \xi(t)\ud t, \quad t \in [0,T],
$$
is an $(\bF,\Q)$-Wiener process and the intensity measure of $N$ under $\Q$ is given by
$$
\nu^{\Q}(\ud \zeta)\ud t=(1+\eta(\zeta ; t))\nu(\ud \zeta)\ud t.
$$
\noindent Moreover, the following assumption:
\begin{ass} \label{internal}
{\em The filtration $\bF$ is the natural filtration of $W$ and $N$; i.e.,
\begin{equation}
 \F_t = \sigma\{ W_s, N(A \times (0,s]), B;\ 0\le s \le t,\ A \in \mathcal Z,\ B \in \N \}, \quad t \in [0,T],
 \end{equation}
where $\N$ is the collection of $\P$-null sets from $\F$},
\end{ass}
\noindent implies that every probability measure  $\Q$ locally equivalent to $\P$ has the structure above.
\end{theorem}
We are now in the position to characterize the martingale measures for the market model considered.
\begin{proposition}
(i) Let $\Q$ be a probability measure locally equivalent  to $\P$ with
$$\left.\frac{\ud \Q }{ \ud \P}\right|_{\F_t} = L_t,$$
where $L$ is given in \eqref{density}, satisfying \eqref{altre} and $ \int_Z |K_1(\zeta;t)| | \eta(\zeta ;t)| \nu(\ud \zeta) < \infty$.
Let us assume
\begin{equation}\label{equi}
\xi(t) \sigma_1(t) + \int_Z K_1( \zeta;t) \eta( \zeta ; t) \nu(\ud \zeta) = r(t) -b_1(t), \quad \ud t \otimes \ud \P-a.e.
\end{equation}
then $\Q$ is a martingale measure.

\noindent (ii) Let $\Q$ be a martingale  measure. Under Assumption  \ref{internal}, $\Q$ has the structure given in Theorem \ref{Gi} and satisfies \eqref{equi}.
\end{proposition}

\begin{proof}
(i) For every $t \in [0,T]$, set $\displaystyle \widetilde{S}^1_t=\frac{S^1_t}{S^0_t}$. Then $\widetilde S^1$ satisfies
\begin{equation} \label{Clarisky}
\begin{aligned}
\ud \widetilde{S}^1_t&=  \widetilde{S}^1_{t^-}\left( [ b_1(t) - r(t) ] \ud t+\sigma_1(t)\ud W_t+\int_Z K_1(\zeta;t)\tilde N(\ud \zeta, \ud t)\right)\\
& = \widetilde{S}^1_{t^-} \bigg( [ b_1(t) - r(t) + \xi(t)  \sigma_1(t) +  \int_Z \eta( \zeta ;t)  K_1(\zeta;t) \nu( \ud \zeta) ] \ud t\\
& \quad + \sigma_1(t) \ud W^{\mathbf Q}_t+ \int_Z K_1(\zeta;t) (N(\ud \zeta, \ud t) - \nu^{\mathbf Q}(\ud \zeta)\ud t) \bigg) \quad \mbox{ for every } \ t \in [0,T],
\end{aligned}
\end{equation}
and condition  \eqref{equi} implies that $\widetilde{S}^1$ is an $(\bF,\Q)$-local martingale.

\noindent (ii) The converse heavily depends on the fact that we have assumed the internal filtration, i.e. Assumption \ref{internal}. This condition implies that  every probability measure  $\Q$ locally equivalent to $\P$ has the structure given in Theorem \ref{Gi} and since $\tilde S^1$ is an $(\bF,\Q)$-local martingale  then \eqref{equi} is fulfilled.
\end{proof}
 \begin{proposition}\label{prop:GOP}
(i) If the processes $\theta_1$ and $\psi_\theta$ are such that for every $t\in [0,T]$
\begin{equation} \label{Ksh}
 \mathbb{E}\left[ \exp\left\{ \int_0^t \theta_1(s)^2 ds + \int_0^t \int_Z \psi_\theta(\zeta; s) ^2 \nu(\ud \zeta) \ud s \right\}\right] < \infty  \end{equation}
 and
 \begin{equation} \label{Ksh1}
 \int_0^t \int_Z| 1-  \psi_\theta(\zeta; s)| \nu(\ud \zeta) \ud s  < \infty \quad \P-a.s.,
 \end{equation}
 then there exists a  martingale measure  $\mathbf Q$ defined  as $$\left.\frac{\ud \mathbf Q }{ \ud \P}\right|_{\F_t} = L_t, \quad \mbox{for every } \ t \in [0,T],$$
where $L$ is given in \eqref{density} with $(\xi, \eta) =(- \theta_1, -\psi_\theta)$ and satisfying \eqref{altre}.

\noindent (ii) Under Assumption \ref{internal}, the existence of a  martingale measure  $\mathbf Q$  implies the existence of the GOP with $( \theta_1, \psi_\theta) = (-\xi, - \eta)$.
 \end{proposition}

\begin{proof}
(i) First let us observe that since $\psi_\theta(\zeta;t) < 1$ for every $t \in [0,T]$ and for every $\zeta \in Z$, then $1 + \eta(\zeta ; t) > 0$ for every $t \in [0,T]$ and every $\zeta \in Z$ and by   \eqref{Ksh} also
$ \esp{\int_0^t \xi(s)^2 \ud s}< \infty$ and $ \esp{\int_0^t \int_Z |\eta(\zeta ; s)| \nu(\ud\zeta)} < \infty$. Hence, $L$ is a nonnegative $(\bF, \P)$-local martingale and by Kazuhiro-Shimbo criterium (see \cite{pp2004}  p.141 and 358) it is a true $(\bF, \P)$-martingale. Finally, equation \eqref{gammarepre} implies \eqref{equi}.
\noindent 

(ii) It is a direct consequence of \eqref{equi}.
\end{proof}

\begin{remark}
Let us notice that if $\theta_1$, and $\psi_\theta$ are bounded processes and in addition $\nu(Z) < \infty$, then conditions \eqref{Ksh} and \eqref{Ksh1} are satisfied, hence there exists a martingale measure for our market.
\end{remark}


\begin{thebibliography}{99}

%\bibitem{b01} Becherer, D. (2001)
 % The num\'{e}raire portfolio for unbounded semimartingales,
 % {\em Finance and Stochastics},
  %Volume 5, Number 3, pages 327-341.

\bibitem{bc} {\sc Biagini, F. and Cretarola, A.} (2009). Local risk-minimization for defaultable markets. {\em Mathematical Finance} {\bf 19} (4) 669-689.


\bibitem{bcp} {\sc  Biagini, F., Cretarola, A. and Platen, E.} (2012). Local risk-minimization under the benchmark approach. Preprint available at: \url{http://arxiv.org/abs/1210.2337}.

%\bibitem{bp} {\sc F. Biagini and M. Pratelli}, {\it   Local risk-minimization and num\'{e}raire}, Journal of Applied Probability, 36 (4) (1999), pp. 1126-1139.

\bibitem{BKR} {\sc Bjork, T., Kabanov, Y. and  Runggaldier, W.} (1997). Bond market structure in presence of
marked point processes. {\em Mathematical Finance} {\bf 7} (2)  211-223.

\bibitem{Br}  {\sc Br\'{e}maud, P.} (1980). {\em Point Processes and Queues}  Springer-Verlag, New York.

\bibitem{C} {\sc Ceci, C.} (2006). Risk minimizing hedging for a partially observed high frequency data
model. {\em Stochastics: An International Journal of Probability and Stochastic Processes} {\bf 78} (1) 13-31.

\bibitem{cco1} {\sc Ceci, C. and Colaneri, K.} (2012). Nonlinear filtering for jump diffusion observations. {\em Advances in Applied Probability} {\bf 44} (3)  678-701.

\bibitem{cco2} {\sc Ceci, C. and Colaneri, K.} (2013).  The Zakai equation of nonlinear filtering for jump-diffusion
  observation: existence and uniqueness. Preprint available at: \url{http://arxiv.org/abs/1210.4279i}.

\bibitem{ccr} {\sc Ceci, C., Cretarola, A. and  Russo, F.} (2012). GKW representation theorem and linear BSDEs under restricted information. An application to risk-minimization. Preprint available at: \url{http://arxiv.org/abs/1205.3726}. %Accepted for pubblication.

\bibitem{ccr2} {\sc Ceci, C., Cretarola, A. and  Russo, F.} (2013). BSDEs under partial information and financial applications. Preprint available at: \url{http://arxiv.org/abs/1305.3690}.

\bibitem{cg} {\sc Ceci, C. and Gerardi, A.} (2006). A model for high frequency data under partial
information: a filtering approach. {\em International Journal of Theoretical and
Applied Finance} {\bf 9} (4) 1-22.

\bibitem{cg1} {\sc  Ceci, C. and Gerardi, A.} (2009). Pricing for geometric marked point processes under partial information: entropy approach.    {\em  International Journal of Theoretical and Applied Finance} {\bf 12} (2)  179-207.

\bibitem{cp2005} {\sc  Christensen, M.~M. and Platen, E.} (2005). A general benchmark model for
stochastic jump sizes. {\em Stochastic Analysis and Applications} {\bf 23} (5) 1017-1044.

\bibitem{dm2} {\sc Dellacherie, C. and Meyer, P.~A. } (1982).
{\em Probabilities and Potential B.} North Holland,
   Amsterdam.

\bibitem{dp} {\sc Du, K. and Platen, E.} (2011). Benchmarked risk-minimization for jump diffusion markets. Preprint, Quantitative Finance Research Centre, Research Paper 296.

%\bibitem{f2002} {\sc E.~R. Fernholz}, {\it Stochastic Portfolio Theory}, Springer-Verlag, New York, 2002.

%\bibitem{fk2009} {\sc E.~R. Fernholz and I. Karatzas}, {\it  Stochastic portfolio theory: an overview}, Handbook of Numerical Analysis,
%Volume XV of {\em Mathematical Modeling and Numerical Methods in Finance}
%(Bensoussan A. and Q. Zhan, eds.). North Holland, Amsterdam, (2009), pp. 89-167.

\bibitem{fs1991} {\sc F\"{o}llmer, H. and Schweizer, M.} (1991). Hedging of contingent claims under incomplete
information. {\em Applied Stochastic Analysis, Stochastic Monographs 5} 389-414.

\bibitem{fs86} {\sc F\"{o}llmer, H. and  Sondermann, D.} (1986). Hedging of Non-redundant Contingent Claims. {\em  Contributions to Mathematical Economics. In Honor of G. Debreu (Hildenbrand W. and A. Mas-Colell eds.)} Elsevier Science Publ.,  North-Holland 205-223.

\bibitem{frey2000} {\sc Frey, R.} (2000). Risk minimization with incomplete information in a model for high-frequency data. {\em Mathematical
Finance} {\bf 10} (2) 215-22.

\bibitem {fr2001}  {\sc Frey, R. and Runggaldier, W.} (2001). A nonlinear filtering approach to volatility estimation with a view towards high
frequency data. {\em International Journal of Theoretical and Applied Finance} {\bf 4} (2) 199-210.

\bibitem{hs} {\sc Hulley, H. and Schweizer, M.} (2010). $M^6$ - On Minimal Market Models and Minimal Martingale Measures.  {\em Contemporary Quantitative Finance, Essays in Honour of Eckhard Platen (C. Chiarella and A. Novikov  eds.)} Springer 35-51.

%\bibitem{js} Jacod, J. and A.~N. Shiryaev (2003)
%  {\em Limit Theorems for Stochastic Processes}.
%  Springer-Verlag Berlin Heidelberg.

\bibitem{kk} {\sc Karatzas, I. and Kardaras, C.} (2007).  The num\'{e}raire portfolio in semimartingale financial models.
{\em    Finance and Stochastics} {\bf 11} (4) 447-493.

\bibitem{kar12} {\sc Kardaras, C.} (2012). Market viability via absence of arbitrage of the first kind. {\em Finance and Stochastics} {\bf 16} (4)  651-667.

\bibitem{k56} {\sc Kelly, J.~L. Jr.} (1956).  A new interpretation of information rate. {\em   Bell System Technical Journal} {\bf 35} (4) 917-926.

\bibitem{KO} {\sc Kurtz, T.~G.  and Ocone, D.} (1988). Unique characterization of conditional distributions in nonlinear filtering. {\em Annals of  Probability} {\bf 16} 80-107.

\bibitem{l90} {\sc Long, J.~B.} (1990).  The num\'{e}raire portfolio.
{\em  Journal of Financial Economics} {\bf 26} (1)  29-69.

\bibitem{m73} {\sc Merton, R.~C.} (1973). An intertemporal capital asset pricing model. {\em Econometrica} {\bf 41} (5) 867-887.

\bibitem{OS} {\sc {\O}ksendal, B. and Sulem, A.} (2007). {\em Applied stochastic control of jump diffusions, $2^{nd}$ edition} Springer.

\bibitem{p2008} {\sc Platen, E.} (2008).  A unifying approach to asset pricing.  Research paper series 227, Quantitative Finance  Research Center, University of Technology, Sydney.

\bibitem{p2005} {\sc Platen, E.} (2005). Diversified portfolios with jumps in a benchmark framework. {\em Asia-Pacific Financial Markets} {\bf 11} (1) 1-22.

\bibitem{ph} {\sc Platen, E. and Heath, D.} (2006). {\em A Benchmark Approach to Quantitative Finance} Springer Finance, Springer-Verlag Berlin Heidelberg.


%\bibitem{cc2011} Ceci, C. and K. Colaneri (2011) Nonlinear Filtering for Jump Diffusion Processes with a Financial Application. Preprint.

%\bibitem{hp97} Pham, H. (1997) Optimal Stopping, Free Boundary, and American Option in a Jump-Diffusion Model. {\em Appl. Math. Optim.}, volume 35, 145-164.


\bibitem{pp2004} {\sc Protter, P.} (2004). {\em Stochastic Integration and Differential Equations, $2^{nd}$ edition} Springer Verlag.


\bibitem{rs2000} {\sc Rydberg, T. and Shephard, N.} (2000).  A modelling framework for prices, trades made at the New York stock exchange.
{\em Nuffield College working paper series 1999-W14, Oxford, Nonlinear, Nonstationary Signal Processing, W.J. Fitzgerald et al. eds.} Cambridge University Press 217-246.

\bibitem{r2002} {\sc Runggaldier, W.} (2002). {\em Jump-Diffusion models} Handbook of
Heavy Tailed Distributions in Finance, S.T. Rachev ed., North Holland Handbooks in Finance.

\bibitem{s94} {\sc Schweizer, M.} (1994).
 Risk-minimizing hedging strategies under restricted information.
{\em Mathematical Finance} {\bf 4} (4) 327-342.

\bibitem{sch01} {\sc Schweizer, M.} (2001). A guided tour through quadratic hedging approaches. {\em Option Pricing, Interest Rate and Risk  Management (Jouini E., Cvitanic J. and Musiela M. eds.)} Cambridge University Press 538-574.

\end{thebibliography}
\end{document}